\documentclass[a4paper,10pt]{article}
\usepackage{amssymb}
\usepackage{amsmath}
\usepackage{epsfig}
\usepackage{subfigure}
\usepackage{graphics}

\usepackage{latexsym}
\usepackage{rotating}
\usepackage{titlesec}
\newcommand{\squeezeup}{\vspace{-9.5mm}}

\title{Shifted nonlocal reductions of $5$-component Maccari system}

\author{Sena Bayl{\i} \thanks{sbayli@bartin.edu.tr}\\
{\small Department of Mathematics, Faculty of Science}\\
{\small Bart{\i}n University, 74100 Bart{\i}n - Turkiye}\\
Asl{\i} Pekcan \thanks{aslipekcan@hacettepe.edu.tr} \\
{\small Department of Mathematics, Faculty of Science} \\
{\small Hacettepe University, 06800 Ankara - Turkiye}
}

\date{\nonumber}
\begin{document}
\maketitle
\date{\nonumber}
\newtheorem{thm}{Theorem}[section]
\newtheorem{proof}{Proof}[section]
\newtheorem{remark}{Remark}[section]
\newtheorem{Le}{Lemma}[section]
\newtheorem{defi}{Definition}[section]
\newtheorem{ex}{Example}[section]
\newtheorem{pro}{Proposition}[section]
\baselineskip 17pt

\numberwithin{equation}{section}

\begin{abstract}
In this work, we prove that shifted nonlocal reductions of integrable $(2+1)$-dimensional $5$-component Maccari system are particular cases of shifted scale transformations. We present all shifted nonlocal reductions of this system and obtain new two-place and four-place integrable systems and equations. In addition to that we use the Hirota direct method and obtain one-soliton solution of the $5$-component Maccari system. By using the reduction formulas with the solution of the Maccari system we also
derive soliton solutions of the shifted nonlocal reduced Maccari systems and equations. We give some particular examples of solutions with their graphs.
\end{abstract}


\tableofcontents

\section{Introduction}

Recent studies have placed significant emphasis on deriving new integrable standard (unshifted) nonlocal and shifted nonlocal nonlinear partial differential equations after the works of Ablowitz and Musslimani \cite{abl1}-\cite{AbMu5}. One of the powerful aspects of the nonlocal reductions is that if a nonlocal reduction is done consistently the reduced equation preserves integrability. Nonlocal unshifted and shifted reductions of many famous integrable systems like nonlinear Schr\"{o}dinger (NLS), Korteweg-de Vries (KdV), modified Korteweg-de Vries (MKdV) had been derived. Their soliton solutions were obtained by various methods like inverse scattering transform, Hirota method, Darboux transformation etc. \cite{gur1}-\cite{shiftedMKdV22022}.

 The general $(2+1)$-dimensional $(N+1)$-component Maccari system \cite{Maccari1997} is given by
\begin{align}
&iu_{k,t}+u_{k,xx}+pu_k=0, \quad k=1,2,\ldots,N,\label{eqn1}\\
&p_y=\sum_{k=1}^N\sigma_k (u_k\bar{u}_k)_x,\label{eqn2}
\end{align}
where $u_k=u_k(x,y,t)$, $\sigma_k=\pm 1$ for $k=1,2,\ldots,N$, and $p=p(x,y,t)$. Here, the bar notation is used for the complex conjugation, and the functions $u_k$ denote $N$ different short-wave amplitudes, while the function $p$ represents the long-wave amplitude. The system (\ref{eqn1})-(\ref{eqn2}) models nonlinear waves occurring in various physical environments, e.g., rogue waves, isolated waves localized in a very small part of space etc. \cite{Maccari2020}. It has applications in plasma physics, nonlinear optics, hydrodynamics, superconductivity, Bose-Einstein condensates, and so on. There are many works studying the different types of solution of this system, see \cite{Uthayakumar}-\cite{ZMH}.

Taking $N=4$ and $\displaystyle t\longrightarrow \frac{a}{i}t$, we have the $(2+1)$-dimensional $5$-component Maccari system
\begin{align}
&au_t+u_{xx}+pu=0,\label{M1}\\
&av_t+v_{xx}+pv=0,\label{M2}\\
&aw_t+w_{xx}+pw=0,\label{M3}\\
&az_t+z_{xx}+pz=0,\label{M4}\\
&p_y=[\sigma_1 u\bar{u}+\sigma_2 v\bar{v}+\sigma_3 w\bar{w}+\sigma_4 z\bar{z}]_x. \label{M5}
\end{align}
If we take $v=u$, $z=w$, we get $3$-component Maccari system which was first derived by Maccari \cite{Maccari1997} from integrable Nizhnik-Novikov-Veselov equation. In fact, this system can be derived from many nonlinear partial differential equations \cite{Maccari2020}. Uthayakumar et al. \cite{Uthayakumar} used Painlev\'{e} analysis to prove the integrability of the $3$-component Maccari system. Pekcan considered the $3$-component Maccari system, its unshifted nonlocal reductions, and soliton solutions
in \cite{pek2021}. The $3$-component Maccari system can only be reduced to two-place nonlocal equations. For example, Pekcan obtained reverse $y$-space unshifted nonlocal Maccari equation
\begin{equation}
au_t(x,y,t)+u_{xx}(x,y,t)+D_y^{-1}\sigma_1u(x,y,t)\bar{u}(x,y,t)+\sigma_2u(x,-y,t)\bar{u}(x,-y,t))_x u(x,y,t)=0,
\end{equation}
which includes two different places; $(x,y,t)$ and $(x,-y,t)$. In real world problems, there may be more than two events which occur at different places but related. Hence,
to model such problems we need multi-place equations \cite{Loumulti}, \cite{Mamulti}.

In this work, we analyze the $5$-component Maccari system (\ref{M1})-(\ref{M5}). We first show that shifted nonlocal reductions of
this system in the form
\begin{align}
&v(x,y,t)=\rho_1 u(\varepsilon_1x+x_0,\varepsilon_2y+y_0,\varepsilon_3t+t_0),\\
&z(x,y,t)=\rho_2 w(\varepsilon_1x+x_0,\varepsilon_2y+y_0,\varepsilon_3t+t_0),
\end{align}
and
\begin{align}
&v(x,y,t)=\rho_1 \bar{u}(\varepsilon_1x+x_0,\varepsilon_2y+y_0,\varepsilon_3t+t_0),\\
&z(x,y,t)=\rho_2 \bar{w}(\varepsilon_1x+x_0,\varepsilon_2y+y_0,\varepsilon_3t+t_0),
\end{align} for $x_0, y_0, t_0\in\mathbb{R},\, \rho_1^2=\rho_2^2=1$, $\varepsilon_j^2=1$, $j=1, 2, 3$, are special type of shifted scale transformations \cite{GPZ}. We then obtain all consistent shifted nonlocal reductions of $5$-component Maccari system. Under these reductions we derive two-place and also four-place shifted nonlocal equations. We use the Hirota bilinear method \cite{Hirota2}, \cite{Hietarinta} to obtain one-soliton solution of the $5$-component Maccari system.
Then using solution of this system with the shifted nonlocal reduction formulas we derive solutions of the reduced shifted nonlocal equations. Here we have two approaches called Type 1 and Type 2. Let $r(x,t)=\frac{U_1(x,t)}{U_2(x,t)}$ and $q(x,t)=\frac{V_1(x,t)}{V_2(x,t)}$. For instance, consider the reduction formula $r(x,t)=\rho q(\varepsilon_1x+x_0,\varepsilon_2t+t_0)$ where $\varepsilon_1^2=\varepsilon_2^2=1$, $\rho, x_0, t_0 \in \mathbb{R}$. We have
\begin{equation}\label{types}
\frac{U_1(x,t)}{U_2(x,t)}=\rho \frac{V_1(\varepsilon_1x+x_0,\varepsilon_2t+t_0)}{V_2(\varepsilon_1x+x_0,\varepsilon_2t+t_0)}.
\end{equation}
Type 1 is based on equating the numerators and denominators of $r(x,t)$ and $q(x,t)$ separately, i.e., from \eqref{types} we have
two equations
\begin{equation}\label{type1}
U_1(x,t)=\rho V_1(\varepsilon_1x+x_0,\varepsilon_2t+t_0),\quad U_2(x,t)=V_2(\varepsilon_1x+x_0,\varepsilon_2t+t_0)
\end{equation}
to be satisfied. Type 2 is the approach based on cross multiplication giving
\begin{equation}\label{type2}
U_1(x,t)V_2(\varepsilon_1x+x_0,\varepsilon_2t+t_0)=\rho U_2(x,t)V_1(\varepsilon_1x+x_0,\varepsilon_2t+t_0).
\end{equation}
The equations \eqref{type1} and \eqref{type2} give constraints on the parameters of the solutions.

The paper is organized as follows. In Section 2, we prove that real and complex shifted nonlocal reductions of $5$-component Maccari system are special type of shifted scale transformations. We give the reduced shifted nonlocal Maccari systems. In Section 3, we apply the shifted nonlocal reductions to the reduced nonlocal Maccari systems and obtain new integrable shifted nonlocal  Maccari equations. In Section 4, we apply the Hirota bilinear method to the $5$-component Maccari system to obtain one-soliton solution of this system. We use one-soliton solution of the $5$-component Maccari system with the reduction formulas and derive soliton solutions of the reduced shifted nonlocal Maccari systems in Section 5. By using these solutions we further derive soliton solutions of the reduced shifted nonlocal Maccari equations in Section 6.

\section{Shifted nonlocal reductions}

In this section, inspired by the work of G\"{u}rses et al. \cite{GPZ} we shall prove that the shifted nonlocal reductions of $5$-component Maccari system are special cases of shifted scale transformations. We also present the reduced shifted nonlocal systems obtained from the $5$-component Maccari system \eqref{M1}-\eqref{M5}. Note that we get the function $p(x,y,t)=D_y^{-1} [\sigma_1 u\bar{u}+\sigma_2 v\bar{v}+\sigma_3 w\bar{w}+\sigma_4 z\bar{z}]_x$ from (\ref{M5}) and insert it into \eqref{M1}-\eqref{M4}. We have
\begin{align}
&au_t+u_{xx}+D_y^{-1} [\sigma_1 u\bar{u}+\sigma_2 v\bar{v}+\sigma_3 w\bar{w}+\sigma_4 z\bar{z}]_xu=0,\label{MM1}\\
&av_t+v_{xx}+D_y^{-1} [\sigma_1 u\bar{u}+\sigma_2 v\bar{v}+\sigma_3 w\bar{w}+\sigma_4 z\bar{z}]_xv=0,\label{MM2}\\
&aw_t+w_{xx}+D_y^{-1} [\sigma_1 u\bar{u}+\sigma_2 v\bar{v}+\sigma_3 w\bar{w}+\sigma_4 z\bar{z}]_xw=0,\label{MM3}\\
&az_t+z_{xx}+D_y^{-1} [\sigma_1 u\bar{u}+\sigma_2 v\bar{v}+\sigma_3 w\bar{w}+\sigma_4 z\bar{z}]_xz=0.\label{MM4}
\end{align}
In the below theorems we will use this form of the $5$-component Maccari system.

\subsection{Real shifted nonlocal reductions for the Maccari system}
\begin{thm} \label{theoremRMac}
For the $5$-component Maccari system \eqref{MM1}-\eqref{MM4}, the real shifted nonlocal reductions are special cases of shifted discrete symmetry transformations, which are particular types of shifted scale transformations.
\end{thm}
\begin{proof}
\noindent Let
 \begin{align}
T_R:&\big(u(x,y,t),v(x,y,t),w(x,y,t),z(x,y,t)\big)\longrightarrow \big(u'(x',y',t'),v'(x',y',t'),w'(x',y',t'),z'(x',y',t')\big) \nonumber\\
& \hspace{3cm} x'=\alpha x+x_0,\quad y'=\gamma y+y_0,\quad t'=\beta t+t_0,\nonumber\\
&\hspace{3cm} u'=\delta_1 v,\quad v'=\delta_2 u,\quad w'=\delta_3z,\quad z'=\delta_4w,\label{T5}
\end{align}
where $\alpha,\beta,\gamma,x_0,y_0,t_0,\delta_i \in \mathbb{R}$ for $i=1,2,3,4 $. Here, the functions $u'(x',y',t')$, $v'(x',y',t')$, $w'(x',y',t')$, $z'(x',y',t')$ also satisfy the $5$-component Maccari system
\begin{align}
&au'_{t'}+u'_{x'x'}+D^{-1}_{y'}[\sigma_1 u'\bar{u}'+\sigma_2 v'\bar{v}'+\sigma_3 w'\bar{w}'+\sigma_4 z'\bar{z}']_{x'}u'=0,\label{M'1}\\
&av'_{t'}+v'_{x'x'}+D^{-1}_{y'}[\sigma_1 u'\bar{u}'+\sigma_2 v'\bar{v}'+\sigma_3 w'\bar{w}'+\sigma_4 z'\bar{z}']_{x'}v'=0,\label{M'2}\\
&aw'_{t'}+w'_{x'x'}+D^{-1}_{y'}[\sigma_1 u'\bar{u}'+\sigma_2 v'\bar{v}'+\sigma_3 w'\bar{w}'+\sigma_4 z'\bar{z}']_{x'}w'=0,\label{M'3}\\
&az'_{t'}+z'_{x'x'}+D^{-1}_{y'}[\sigma_1 u'\bar{u}'+\sigma_2 v'\bar{v}'+\sigma_3 w'\bar{w}'+\sigma_4 z'\bar{z}']_{x'}z'=0.\label{M'4}
\end{align}
\noindent Under the scale transformation \eqref{T5}, we obtain the following equations:
\begin{align}
&\frac{a\delta_1}{\beta} v_t+\frac{\delta_1}{\alpha^2}v_{xx}+\frac{\gamma}{\alpha}D^{-1}_{y}[\sigma_1\delta^2_1v\bar{v}+\sigma_2\delta^2_2u\bar{u}+\sigma_3\delta^2_3z\bar{z}+\sigma_4\delta^2_4w\bar{w}]_x\delta_1v=0, \label{RMac1} \\
&\frac{a\delta_2}{\beta} u_t+\frac{\delta_2}{\alpha^2}u_{xx}+\frac{\gamma}{\alpha}D^{-1}_{y}[\sigma_1\delta^2_1v\bar{v}+\sigma_2\delta^2_2u\bar{u}+\sigma_3\delta^2_3z\bar{z}+\sigma_4\delta^2_4w\bar{w}]_x\delta_2u=0, \label{RMac2} \\
&\frac{a\delta_3}{\beta} z_t+\frac{\delta_3}{\alpha^2}z_{xx}+\frac{\gamma}{\alpha}D^{-1}_{y}[\sigma_1\delta^2_1v\bar{v}+\sigma_2\delta^2_2u\bar{u}+\sigma_3\delta^2_3z\bar{z}+\sigma_4\delta^2_4w\bar{w}]_x\delta_3z=0, \label{RMac3} \\
&\frac{a\delta_4}{\beta} w_t+\frac{\delta_4}{\alpha^2}w_{xx}+\frac{\gamma}{\alpha}D^{-1}_{y}[\sigma_1\delta^2_1v\bar{v}+\sigma_2\delta^2_2u\bar{u}+\sigma_3\delta^2_3z\bar{z}+\sigma_4\delta^2_4w\bar{w}]_x\delta_4w=0. \label{RMac4}
\end{align}
\noindent The above system \eqref{RMac1}-\eqref{RMac4} is invariant if
\begin{equation}
\frac{a\alpha^2}{\beta}=a,\quad \gamma \alpha \sigma_2\delta_2^2=\sigma_1, \quad \gamma\alpha\sigma_1\delta_1^2=\sigma_2, \quad \gamma\alpha\sigma_4\delta_4^2=\sigma_3, \quad \gamma\alpha\sigma_3\delta_3^2=\sigma_4,
\end{equation}
yielding $\displaystyle \delta_1^2 \delta_2^2=\delta_3^2\delta_4^2=\frac{1}{\gamma^2\alpha^2}$. Now, let $\gamma=\varepsilon_2=\pm 1, \alpha=\varepsilon_1=\pm 1,$ giving $\beta=1$. Taking $u'=u$, $v'=v$, $w'=w$, $z'=z$, we get the real shifted nonlocal reductions
\begin{align}
&v(x,y,t)=\rho_1 u(\varepsilon_1x+x_0,\varepsilon_2y+y_0,t+t_0),\label{realred1}\\
&z(x,y,t)=\rho_2 w(\varepsilon_1x+x_0,\varepsilon_2y+y_0,t+t_0),\label{realred2}
\end{align}
for $x_0,y_0,t_0\in\mathbb{R},\, \rho_1=\pm 1, \rho_2=\pm 1$.
\hfill$\Box$
\end{proof}

\noindent Using the reductions (\ref{realred1}) and (\ref{realred2}) on the system \eqref{MM1}-\eqref{MM4} and letting
\begin{align}
\Tilde{x}=\varepsilon_1x+x_0,\quad \Tilde{y}=\varepsilon_2y+y_0,\quad \Tilde{t}=t+t_0,\label{tildetfRMac}
\end{align}
we have $\sigma_2=\varepsilon_1\varepsilon_2\sigma_1$, $\sigma_4=\varepsilon_1\varepsilon_2\sigma_3$, and the following three possibilities:
\begin{align} \label{RMacpos}
    &i)\,\, \varepsilon_1=\varepsilon_2=-1,\, t_0=0,\nonumber\\
    &ii)\,\, \varepsilon_1=-1,\, \varepsilon_2=1,\, y_0=t_0=0,\nonumber\\
    &iii)\,\, \varepsilon_1=1,\, \varepsilon_2=-1,\, x_0=t_0=0,\nonumber\\
    &
\end{align}
for consistency. Therefore, we obtain the three different two-place shifted nonlocal Maccari systems below:

\noindent 1)\, Real reverse $x$-space shifted nonlocal Maccari system
\begin{align}
&au_t(x,y,t)+u_{xx}(x,y,t)+D^{-1}_{y}[\sigma_1u(x,y,t)\bar{u}(x,y,t)-\sigma_1u(-x+x_0,y,t)\bar{u}(-x+x_0,y,t)\nonumber \\  &+\sigma_3w(x,y,t)\bar{w}(x,y,t)-\sigma_3w(-x+x_0,y,t)\bar{w}(-x+x_0,y,t)]_xu(x,y,t)=0, \label{RMacx1}\\
&aw_t(x,y,t)+w_{xx}(x,y,t)+D^{-1}_{y}[\sigma_1u(x,y,t)\bar{u}(x,y,t)-\sigma_1u(-x+x_0,y,t)\bar{u}(-x+x_0,y,t) \nonumber \\ &+\sigma_3w(x,y,t)\bar{w}(x,y,t)-\sigma_3w(-x+x_0,y,t)\bar{w}(-x+x_0,y,t)]_xw(x,y,t)=0. \label{RMacx2}
\end{align}
\noindent 2)\, Real reverse $y$-space shifted nonlocal Maccari system
\begin{align}
&au_t(x,y,t)+u_{xx}(x,y,t)+D^{-1}_{y}[\sigma_1u(x,y,t)\bar{u}(x,y,t)-\sigma_1u(x,-y+y_0,t)\bar{u}(x,-y+y_0,t)\nonumber\\ &+\sigma_3w(x,y,t)\bar{w}(x,y,t)-\sigma_3w(x,-y+y_0,t)\bar{w}(x,-y+y_0,t)]_xu(x,y,t)=0, \label{RMacy1}\\
&aw_t(x,y,t)+w_{xx}(x,y,t)+D^{-1}_{y}[\sigma_1u(x,y,t)\bar{u}(x,y,t)-\sigma_1u(x,-y+y_0,t)\bar{u}(x,-y+y_0,t)\nonumber\\ &+\sigma_3w(x,y,t)\bar{w}(x,y,t)-\sigma_3w(x,-y+y_0,t)\bar{w}(x,-y+y_0,t)]_xw(x,y,t)=0. \label{RMacy2}
\end{align}
 \noindent 3)\, Real reverse $xy$-space shifted nonlocal Maccari system
 \begin{align} 
&au_t(x,y,t)+u_{xx}(x,y,t)+D^{-1}_{y}[\sigma_1u(x,y,t)\bar{u}(x,y,t)+\sigma_1u(-x+x_0,-y+y_0,t)\bar{u}(-x+x_0,-y+y_0,t)
\nonumber\\ &+\sigma_3w(x,y,t)\bar{w}(x,y,t)+\sigma_3w(-x+x_0,-y+y_0,t)\bar{w}(-x+x_0,-y+y_0,t)]_xu(x,y,t)=0,\label{RMacxy1} \\
&aw_t(x,y,t)+w_{xx}(x,y,t)+D^{-1}_{y}[\sigma_1u(x,y,t)\bar{u}(x,y,t)+\sigma_1u(-x+x_0,-y+y_0,t)\bar{u}(-x+x_0,-y+y_0,t)
\nonumber\\&+\sigma_3w(x,y,t)\bar{w}(x,y,t)+\sigma_3w(-x+x_0,-y+y_0,t)\bar{w}(-x+x_0,-y+y_0,t)]_xw(x,y,t)=0. \label{RMacxy2}
\end{align}

\subsection{Complex shifted nonlocal reductions for the Maccari system}
\begin{thm} \label{theoremCMac}
For the $5$-component Maccari system \eqref{MM1}-\eqref{MM4}, the complex shifted nonlocal reductions are special cases of shifted discrete symmetry transformations which are particular forms of scale transformations.
\end{thm}
\begin{proof}
\noindent Let
\begin{align}
T_C:&\big(\bar{u}(x,y,t),\bar{v}(x,y,t),\bar{w}(x,y,t),\bar{z}(x,y,t)\big)\longrightarrow \big(u'(x',y',t'),v'(x',y',t'),w'(x',y',t'),z'(x',y',t')\big) \nonumber\\
& \hspace{3cm} x'=\alpha x+x_0,\quad y'=\gamma y+y_0,\quad t'=\beta t+t_0,\nonumber\\
&\hspace{3cm} u'=\delta_1 \bar{v},\quad v'=\delta_2 \bar{u},\quad w'=\delta_3\bar{z},\quad z'=\delta_4\bar{w},\label{T6}
\end{align}
\noindent where $\alpha,\beta,\gamma,x_0, y_0, t_0, \delta_i \in \mathbb{R}$ for $i=1,2,3,4$. The functions $u'(x',y',t'), v'(x',y',t'), w'(x',y',t'), z'(x',y',t')$ satisfy the primed system \eqref{M'1}-\eqref{M'4}. Under the scale transformation \eqref{T6}, the system \eqref{M'1}-\eqref{M'4} turns to be \begin{align}
&\frac{\bar{a}\alpha^2}{\beta} v_t+v_{xx}+\alpha\gamma D^{-1}_{y}[\sigma_1\delta^2_1v\bar{v}+\sigma_2\delta^2_2u\bar{u}+\sigma_3\delta^2_3z\bar{z}+\sigma_4\delta^2_4w\bar{w}]_xv=0,\label{CMac1}\\
&\frac{\bar{a}\alpha^2}{\beta} u_t+u_{xx}+\alpha\gamma D^{-1}_{y}[\sigma_1\delta^2_1v\bar{v}+\sigma_2\delta^2_2u\bar{u}+\sigma_3\delta^2_3z\bar{z}+\sigma_4\delta^2_4w\bar{w}]_xu=0,\label{CMac2}\\
&\frac{\bar{a}\alpha^2}{\beta} z_t+z_{xx}+\alpha\gamma D^{-1}_{y}[\sigma_1\delta^2_1v\bar{v}+\sigma_2\delta^2_2u\bar{u}+\sigma_3\delta^2_3z\bar{z}+\sigma_4\delta^2_4w\bar{w}]_xz=0,\label{CMac3}\\
&\frac{\bar{a}\alpha^2}{\beta} w_t+w_{xx}+\alpha\gamma D^{-1}_{y}[\sigma_1\delta^2_1v\bar{v}+\sigma_2\delta^2_2u\bar{u}+\sigma_3\delta^2_3z\bar{z}+\sigma_4\delta^2_4w\bar{w}]_xw=0.\label{CMac4}
\end{align}
\noindent In order to leave the above system invariant we have
\begin{equation}
\frac{\bar{a}\alpha^2}{\beta}=a,\quad \gamma \alpha \sigma_2\delta_2^2=\sigma_1, \quad \gamma\alpha\sigma_1\delta_1^2=\sigma_2, \quad \gamma\alpha\sigma_4\delta_4^2=\sigma_3, \quad \gamma\alpha\sigma_3\delta_3^2=\sigma_4,
\end{equation}
yielding $\displaystyle \delta_1^2 \delta_2^2=\delta_3^2\delta_4^2=\frac{1}{\gamma^2\alpha^2}$. Furthermore, let $\alpha=\varepsilon_1=\pm 1,\,\gamma=\varepsilon_2=\pm 1,\,\beta=\varepsilon_3=\pm 1$. Hence $a=\bar{a}\varepsilon_3$. Taking $u'=u, v'=v, w'=w, z'=z$ we get the complex shifted nonlocal reductions as
\begin{align}
&v(x,y,t)=\rho_1 \bar{u}(\varepsilon_1x+x_0,\varepsilon_2y+y_0,\varepsilon_3t+t_0),\\
&z(x,y,t)=\rho_2 \bar{w}(\varepsilon_1x+x_0,\varepsilon_2y+y_0,\varepsilon_3t+t_0),
\end{align} for $x_0, y_0, t_0\in\mathbb{R},\, \rho_1, \rho_2=\pm 1$.
\hfill$\Box$
\end{proof}

\noindent By following the similar steps as in the case of real shifted nonlocal reductions we get the constraints
\begin{equation}
a=\bar{a}\varepsilon_3,\quad \sigma_1=\varepsilon_1\varepsilon_2\sigma_2,\quad \sigma_3=\varepsilon_1\varepsilon_2\sigma_4
\end{equation}
to reduce the $5$-component Maccari system \eqref{MM1}-\eqref{MM4} to shifted nonlocal Maccari systems, consistently. Hence we have the following possibilities:
\begin{align}
    &i)\, \varepsilon_1=1,\, \varepsilon_2=1,\, \varepsilon_3=-1,\, x_0=y_0=0,\nonumber\\
    &ii)\, \varepsilon_1=1,\, \varepsilon_2=-1,\, \varepsilon_3=1,\, x_0=t_0=0,\nonumber\\
    &iii)\, \varepsilon_1=1,\, \varepsilon_2=-1,\, \varepsilon_3=-1,\, x_0=0,\nonumber\\
    &iv)\, \varepsilon_1=-1,\, \varepsilon_2=1,\, \varepsilon_3=1,\, y_0=t_0=0,\nonumber\\
    &v)\, \varepsilon_1=-1,\, \varepsilon_2=1,\, \varepsilon_3=-1,\, y_0=0,\nonumber\\
    &vi)\, \varepsilon_1=-1,\, \varepsilon_2=-1,\, \varepsilon_3=1,\, t_0=0,\nonumber\\
    &vii)\, \varepsilon_1=-1,\, \varepsilon_2=-1,\, \varepsilon_3=-1.\nonumber\\
    &\label{CMacpos}
\end{align}
Explicitly, we obtain the following seven two-place complex shifted nonlocal Maccari systems:

\noindent 1)\, Complex reverse $x$-space shifted nonlocal Maccari system 
\begin{align}
&au_t(x,y,t)+u_{xx}(x,y,t)+D^{-1}_{y}[\sigma_1u(x,y,t)\bar{u}(x,y,t)-\sigma_1u(-x+x_0,y,t)\bar{u}(-x+x_0,y,t)\nonumber \\  &+\sigma_3w(x,y,t)\bar{w}(x,y,t)-\sigma_3w(-x+x_0,y,t)\bar{w}(-x+x_0,y,t)]_xu(x,y,t)=0, \label{CMacx1}\\
&aw_t(x,y,t)+w_{xx}(x,y,t)+D^{-1}_{y}[\sigma_1u(x,y,t)\bar{u}(x,y,t)-\sigma_1u(-x+x_0,y,t)\bar{u}(-x+x_0,y,t) \nonumber \\ &+\sigma_3w(x,y,t)\bar{w}(x,y,t)-\sigma_3w(-x+x_0,y,t)\bar{w}(-x+x_0,y,t)]_xw(x,y,t)=0, \label{CMacx2}
\end{align}
where $a=\bar{a}$.

\noindent 2)\, Complex reverse $y$-space shifted nonlocal Maccari system 
\begin{align}
&au_t(x,y,t)+u_{xx}(x,y,t)+D^{-1}_{y}[\sigma_1u(x,y,t)\bar{u}(x,y,t)-\sigma_1u(x,-y+y_0,t)\bar{u}(x,-y+y_0,t)\nonumber\\ &+\sigma_3w(x,y,t)\bar{w}(x,y,t)-\sigma_3w(x,-y+y_0,t)\bar{w}(x,-y+y_0,t)]_xu(x,y,t)=0, \label{CMacy1}\\
&aw_t(x,y,t)+w_{xx}(x,y,t)+D^{-1}_{y}[\sigma_1u(x,y,t)\bar{u}(x,y,t)-\sigma_1u(x,-y+y_0,t)\bar{u}(x,-y+y_0,t)\nonumber\\ &+\sigma_3w(x,y,t)\bar{w}(x,y,t)-\sigma_3w(x,-y+y_0,t)\bar{w}(x,-y+y_0,t)]_xw(x,y,t)=0, \label{CMacy2}
\end{align}
where $a=\bar{a}$.

\noindent 3)\, Complex reverse $xy$-space shifted nonlocal Maccari system 
\begin{align}
&au_t(x,y,t)+u_{xx}(x,y,t)+D^{-1}_{y}[\sigma_1u(x,y,t)\bar{u}(x,y,t)+\sigma_1u(-x+x_0,-y+y_0,t)\nonumber\\ &\times\bar{u}(-x+x_0,-y+y_0,t)+\sigma_3w(x,y,t)\bar{w}(x,y,t)+\sigma_3w(-x+x_0,-y+y_0,t)\nonumber\\ &\times\bar{w}(-x+x_0,-y+y_0,t)]_xu(x,y,t)=0,\label{CMacxy1} \\
&aw_t(x,y,t)+w_{xx}(x,y,t)+D^{-1}_{y}[\sigma_1u(x,y,t)\bar{u}(x,y,t)+\sigma_1u(-x+x_0,-y+y_0,t)\nonumber\\ &\times\bar{u}(-x+x_0,-y+y_0,t)+\sigma_3w(x,y,t)\bar{w}(x,y,t)+\sigma_3w(-x+x_0,-y+y_0,t)\nonumber\\ &\times\bar{w}(-x+x_0,-y+y_0,t)]_xw(x,y,t)=0, \label{CMacxy2}
\end{align}
where $a=\bar{a}$.

\noindent 4)\, Complex reverse time shifted nonlocal Maccari system
\begin{align}
&au_t(x,y,t)+u_{xx}(x,y,t)+D^{-1}_{y}[\sigma_1u(x,y,t)\bar{u}(x,y,t)+\sigma_1u(x,y,-t+t_0)\bar{u}(x,y,-t+t_0)\nonumber\\ &+\sigma_3w(x,y,t)\bar{w}(x,y,t)+\sigma_3w(x,y,-t+t_0)\bar{w}(x,y,-t+t_0)]_xu(x,y,t)=0,\label{CMact1} \\
&aw_t(x,y,t)+w_{xx}(x,y,t)+D^{-1}_{y}[\sigma_1u(x,y,t)\bar{u}(x,y,t)+\sigma_1u(x,y,-t+t_0)\bar{u}(x,y,-t+t_0)\nonumber\\ &+\sigma_3w(x,y,t)\bar{w}(x,y,t)+\sigma_3w(x,y,-t+t_0)\bar{w}(x,y,-t+t_0)]_xw(x,y,t)=0, \label{CMact2}
\end{align}
where $a=-\bar{a}$.

\noindent 5)\, Complex reverse $x$-space-time shifted nonlocal Maccari system 
\begin{align}
&au_t(x,y,t)+u_{xx}(x,y,t)+D^{-1}_{y}[\sigma_1u(x,y,t)\bar{u}(x,y,t)-\sigma_1u(-x+x_0,y,-t+t_0)\bar{u}(-x+x_0,y,-t+t_0)\nonumber \\  &+\sigma_3w(x,y,t)\bar{w}(x,y,t)-\sigma_3w(-x+x_0,y,-t+t_0)\bar{w}(-x+x_0,y,-t+t_0)]_xu(x,y,t)=0, \label{CMacxt1}\\
&aw_t(x,y,t)+w_{xx}(x,y,t)+D^{-1}_{y}[\sigma_1u(x,y,t)\bar{u}(x,y,t)-\sigma_1u(-x+x_0,y,-t+t_0)\bar{u}(-x+x_0,y,-t+t_0) \nonumber \\ &+\sigma_3w(x,y,t)\bar{w}(x,y,t)-\sigma_3w(-x+x_0,y,-t+t_0)\bar{w}(-x+x_0,y,-t+t_0)]_xw(x,y,t)=0, \label{CMacxt2}
\end{align}
where $a=-\bar{a}$.

\noindent 6)\, Complex reverse $y$-space-time shifted nonlocal Maccari system 
\begin{align}
&au_t(x,y,t)+u_{xx}(x,y,t)+D^{-1}_{y}[\sigma_1u(x,y,t)\bar{u}(x,y,t)-\sigma_1u(x,-y+y_0,-t+t_0)\bar{u}(x,-y+y_0,-t+t_0)\nonumber\\ &+\sigma_3w(x,y,t)\bar{w}(x,y,t)-\sigma_3w(x,-y+y_0,-t+t_0)\bar{w}(x,-y+y_0,-t+t_0)]_xu(x,y,t)=0, \label{CMacyt1}\\
&aw_t(x,y,t)+w_{xx}(x,y,t)+D^{-1}_{y}[\sigma_1u(x,y,t)\bar{u}(x,y,t)-\sigma_1u(x,-y+y_0,-t+t_0)\bar{u}(x,-y+y_0,-t+t_0)\nonumber\\ &+\sigma_3w(x,y,t)\bar{w}(x,y,t)-\sigma_3w(x,-y+y_0,-t+t_0)\bar{w}(x,-y+y_0,-t+t_0)]_xw(x,y,t)=0, \label{CMacyt2}
\end{align}
where $a=-\bar{a}$.

\noindent 7)\, Complex reverse $xy$-space-time shifted nonlocal Maccari system 
\begin{align}
&au_t(x,y,t)+u_{xx}(x,y,t)+D^{-1}_{y}[\sigma_1u(x,y,t)\bar{u}(x,y,t)+\sigma_1u(-x+x_0,-y+y_0,-t+t_0)\nonumber\\ &\times\bar{u}(-x+x_0,-y+y_0,-t+t_0)+\sigma_3w(x,y,t)\bar{w}(x,y,t)+\sigma_3w(-x+x_0,-y+y_0,-t+t_0)\nonumber\\ &\times\bar{w}(-x+x_0,-y+y_0,-t+t_0)]_xu(x,y,t)=0,\label{CMacxyt1} \\
&aw_t(x,y,t)+w_{xx}(x,y,t)+D^{-1}_{y}[\sigma_1u(x,y,t)\bar{u}(x,y,t)+\sigma_1u(-x+x_0,-y+y_0,-t+t_0)\nonumber\\ &\times\bar{u}(-x+x_0,-y+y_0,-t+t_0)+\sigma_3w(x,y,t)\bar{w}(x,y,t)+\sigma_3w(-x+x_0,-y+y_0,-t+t_0)\nonumber\\ &\times\bar{w}(-x+x_0,-y+y_0,-t+t_0)]_xw(x,y,t)=0, \label{CMacxyt2}
\end{align}
where $a=-\bar{a}$.

\textbf{Remark.} Note that if $a=\bar{a}$ then the real and complex reverse $x$-space shifted nonlocal systems are the same. A similar result holds for the reverse $y$-space and $xy$-space shifted nonlocal systems.

\section{Shifted nonlocal Maccari equations}

We can further reduce the real and complex shifted nonlocal Maccari systems to shifted nonlocal Maccari equations by using the following reductions:
\begin{align}
 &\textbf{I.}\,  w(x,y,t)=\rho_3u(\mu_1x+X_0,\mu_2y+Y_0,\mu_3t+T_0),\,\, \rho_3^2=\mu_j^2=1, j=1, 2, 3,\,\, X_0, Y_0, T_0 \in \mathbb{R}, \label{rho3realtf}\\
 &\textbf{II.}\,  w(x,y,t)=\rho_3\bar{u}(\mu_1x+X_0,\mu_2y+Y_0,\mu_3t+T_0),\,\, \rho_3^2=\mu_j^2=1, j=1, 2, 3,\,\, X_0, Y_0, T_0 \in \mathbb{R}. \label{rho3complextf}
\end{align}

\subsection{Shifted nonlocal Maccari equations obtained from the real shifted nonlocal Maccari systems}

\noindent \textbf{A.}\, Let us first apply the reduction (\ref{rho3realtf}) to the real shifted nonlocal Maccari systems.

\noindent \textbf{A.I.} For consistent reduction of the real $x$-space shifted nonlocal Maccari system \eqref{RMacx1}-\eqref{RMacx2} under (\ref{rho3realtf}),  we get the constraints
 \begin{align}
\mu_3=1,\quad T_0=0, \quad \sigma_3=\mu_1\mu_2\sigma_1, \label{A1con}
\end{align}
and
\begin{align}
    &i)\, \mu_1=-1,\, X_0=x_0; \, \mu_1=1,\, X_0=0,\nonumber\\
    &ii)\, \mu_2=1,\, Y_0=0.
    \end{align}
Hence the system  \eqref{RMacx1}-\eqref{RMacx2} reduces to the following shifted nonlocal Maccari equations:
 \begin{align}
&1)\, au_t(x,y,t)+u_{xx}(x,y,t)+2\sigma_1D^{-1}_{y}[u(x,y,t)\bar{u}(x,y,t)-u(-x+x_0,y,t)\bar{u}(-x+x_0,y,t)]_x u(x,y,t)=0,\label{RRx1} \\
&2)\, au_t(x,y,t)+u_{xx}(x,y,t)+\sigma_1D^{-1}_{y}[u(x,y,t)\bar{u}(x,y,t)-u(-x+x_0,y,t)\bar{u}(-x+x_0,y,t)\nonumber\\  &-u(x,-y+Y_0,t)\bar{u}(x,-y+Y_0,t)+u(-x+x_0,-y+Y_0,t)\bar{u}(-x+x_0,-y+Y_0,t)]_xu(x,y,t)=0. \label{RRx2}
\end{align}
While the equation (\ref{RRx1}) is a two-place equation, the equation (\ref{RRx2}) is four-place.

\noindent \textbf{A.II.} To reduce the real $y$-space shifted nonlocal Maccari system \eqref{RMacy1}-\eqref{RMacy2} by (\ref{rho3realtf}) consistently we get the conditions (\ref{A1con}) and
\begin{align}
    &i)\, \mu_1=1,\, X_0=0,\nonumber\\
    &ii)\, \mu_2=-1,\, Y_0=y_0; \, \mu_2=1,\, Y_0=0.
    \end{align}
Therefore, the system  \eqref{RMacy1}-\eqref{RMacy2} reduces to the following shifted nonlocal Maccari equations:
\begin{align}
&1)\, au_t(x,y,t)+u_{xx}(x,y,t)+2\sigma_1D^{-1}_{y}[u(x,y,t)\bar{u}(x,y,t)-u(x,-y+y_0,t)\bar{u}(x,-y+y_0,t)]_xu(x,y,t)=0. \label{RRy1}\\
&2)\, au_t(x,y,t)+u_{xx}(x,y,t)+\sigma_1D^{-1}_{y}[u(x,y,t)\bar{u}(x,y,t)-u(x,-y+y_0,t)\bar{u}(x,-y+y_0,t)\nonumber\\ &-u(-x+X_0,y,t)\bar{u}(-x+X_0,y,t)+u(-x+X_0,-y+y_0,t)\bar{u}(-x+X_0,-y+y_0,t)]_x u(x,y,t)=0. \label{RRy2}
\end{align}
Here the equation (\ref{RRy1}) is a two-place equation, but (\ref{RRy2}) is a four-place equation.

\noindent \textbf{A.III.} For consistent reduction of the real $xy$-space shifted nonlocal Maccari system \eqref{RMacxy1}-\eqref{RMacxy2} under (\ref{rho3realtf}) we get the conditions (\ref{A1con}) and
\begin{align}
    &i)\, \mu_1=-1,\, X_0=x_0; \, \mu_1=1,\, X_0=0,\nonumber\\
    &ii)\, \mu_2=-1,\, Y_0=y_0; \, \mu_2=1,\, Y_0=0.
\end{align}
Hence the system  \eqref{RMacxy1}-\eqref{RMacxy2} reduces to the following shifted nonlocal Maccari equations:
\begin{align}
&1)\, au_t(x,y,t)+u_{xx}(x,y,t)+2\sigma_1D^{-1}_{y}[u(x,y,t)\bar{u}(x,y,t)+u(-x+x_0,-y+y_0,t)\bar{u}(-x+x_0,-y+y_0,t)]_x\nonumber\\
&\times u(x,y,t)=0, \label{RRxy1}\\
&2)\, au_t(x,y,t)+u_{xx}(x,y,t)+\sigma_1D^{-1}_{y}[u(x,y,t)\bar{u}(x,y,t)-u(-x+x_0,y,t)\bar{u}(-x+x_0,y,t)\nonumber\\
&-u(x,-y+y_0,t)\bar{u}(x,-y+y_0,t)+u(-x+x_0,-y+y_0,t)\bar{u}(-x+x_0,-y+y_0,t)]_x u(x,y,t)=0. \label{RRxy2}
\end{align}
The above first equation is two-place while the second equation is a four-place shifted nonlocal Maccari equation.

\noindent \textbf{B.}\, Now we shall consider the complex shifted reduction (\ref{rho3complextf}) and apply it on the real shifted nonlocal Maccari systems.

\noindent \textbf{B.I.}\, For consistent reduction of the real $x$-space shifted nonlocal Maccari system \eqref{RMacx1}-\eqref{RMacx2} under (\ref{rho3complextf}), we get the constraints as follows:
\begin{equation}
a=\bar{a}\mu_3,\quad \sigma_3=\mu_1\mu_2\sigma_1, \label{B1con}
\end{equation}
and
\begin{align}
    &i)\, \mu_1=-1,\, X_0=x_0; \, \mu_1=1,\, X_0=0,\nonumber\\
    &ii)\, \mu_2=1,\, Y_0=0,\nonumber\\
    &iii)\, \mu_3=1,\, T_0=0.
\end{align}
 \noindent Here we obtain the following four different shifted nonlocal Maccari equations:
 \begin{align}
&1)\, au_t(x,y,t)+u_{xx}(x,y,t)+2\sigma_1D^{-1}_{y}[u(x,y,t)\bar{u}(x,y,t)-u(-x+x_0,y,t)\bar{u}(-x+x_0,y,t)]_x\nonumber\\
&\times u(x,y,t)=0,\, a=\bar{a}, \label{B1.1}\\
&2)\,au_t(x,y,t)+u_{xx}(x,y,t)+\sigma_1D^{-1}_{y}[u(x,y,t)\bar{u}(x,y,t)-u(-x+x_0,y,t)\bar{u}(-x+x_0,y,t)\nonumber \\  &-u(x,-y+Y_0,t)\bar{u}(x,-y+Y_0,t)+u(-x+x_0,-y+Y_0,t)\bar{u}(-x+x_0,-y+Y_0,t)]_xu(x,y,t)=0,\, a=\bar{a}, \label{B1.2}\\
&3)\,au_t(x,y,t)+u_{xx}(x,y,t)+\sigma_1D^{-1}_{y}[u(x,y,t)\bar{u}(x,y,t)-u(-x+x_0,y,t)\bar{u}(-x+x_0,y,t)\nonumber \\  &+u(x,y,-t+T_0)\bar{u}(x,y,-t+T_0)-u(-x+x_0,y,-t+T_0)\bar{u}(-x+x_0,y,-t+T_0)]_x u(x,y,t)=0,\, a=-\bar{a}, \label{CRx3}\\
&4)\,au_t(x,y,t)+u_{xx}(x,y,t)+\sigma_1D^{-1}_{y}[u(x,y,t)\bar{u}(x,y,t)-u(-x+x_0,y,t)\bar{u}(-x+x_0,y,t)\nonumber \\  &-u(x,-y+Y_0,-t+T_0)\bar{u}(x,-y+Y_0,-t+T_0)+u(-x+x_0,-y+Y_0,-t+T_0)\nonumber\\
&\times\bar{u}(-x+x_0,-y+Y_0,-t+T_0)]_xu(x,y,t)=0,\, a=-\bar{a}. \label{CRx4}
\end{align}
The first equation above is a two-place equation but the others are four-place equations.

\noindent \textbf{B.II.}\, To reduce the real $y$-space shifted nonlocal Maccari system \eqref{RMacy1}-\eqref{RMacy2} under the reduction (\ref{rho3complextf}) consistently, we get the conditions (\ref{B1con}) and
\begin{align}
    &i)\, \mu_1=1,\, X_0=0,\nonumber\\
    &ii)\, \mu_2=-1,\, Y_0=y_0\,;\, \mu_2=1,\, Y_0=0,\nonumber\\
    &iii)\, \mu_3=1,\, T_0=0.
\end{align}
Hence we obtain four different shifted nonlocal Maccari equations:
 \begin{align}
&1)\, au_t(x,y,t)+u_{xx}(x,y,t)+2\sigma_1D^{-1}_{y}[u(x,y,t)\bar{u}(x,y,t)-u(x,-y+y_0,t)\bar{u}(x,-y+y_0,t)]_x\nonumber \\
&\times u(x,y,t)=0,\,\, a=\bar{a}, \label{B2.1}\\
&2)\, au_t(x,y,t)+u_{xx}(x,y,t)+\sigma_1D^{-1}_{y}[u(x,y,t)\bar{u}(x,y,t)-u(x,-y+y_0,t)\bar{u}(x,-y+y_0,t)\nonumber \\  &-u(-x+X_0,y,t)\bar{u}(-x+X_0,y,t)+u(-x+X_0,-y+y_0,t)\bar{u}(-x+X_0,-y+y_0,t)]_x\nonumber \\  &\times u(x,y,t)=0,\,\, a=\bar{a}, \label{B2.2}\\
&3)\, au_t(x,y,t)+u_{xx}(x,y,t)+\sigma_1D^{-1}_{y}[u(x,y,t)\bar{u}(x,y,t)-u(x,-y+y_0,t)\bar{u}(x,-y+y_0,t)\nonumber \\  &+u(x,y,-t+T_0)\bar{u}(x,y,-t+T_0)-u(x,-y+y_0,-t+T_0)\bar{u}(x,-y+y_0,-t+T_0)]_x\nonumber \\  &\times u(x,y,t)=0,\,\, a=-\bar{a}, \label{CRy3}\\
&4)\,au_t(x,y,t)+u_{xx}(x,y,t)+\sigma_1D^{-1}_{y}[u(x,y,t)\bar{u}(x,y,t)-u(x,-y+y_0,t)\bar{u}(x,-y+y_0,t)\nonumber \\  &-u(-x+X_0,y,-t+T_0)\bar{u}(-x+X_0,y,-t+T_0)+u(-x+X_0,-y+y_0,-t+T_0)\nonumber \\  &\times \bar{u}(-x+X_0,-y+y_0,-t+T_0)]_xu(x,y,t)=0,\,\, a=-\bar{a}. \label{CRy4}
\end{align}
\noindent The first equation above is a two-place equation but the others are four-place equations.

\noindent \textbf{B.III.}\, Applying the reduction (\ref{rho3complextf}) to the real $xy$-space shifted nonlocal Maccari system \eqref{RMacxy1}-\eqref{RMacxy2} gives the conditions (\ref{B1con}) and
\begin{align}
    &i)\, \mu_1=-1,\, X_0=x_0\,;\, \mu_1=1,\, X_0=0,\nonumber\\
    &ii)\, \mu_2=-1,\, Y_0=y_0\,;\, \mu_2=1,\, Y_0=0,\nonumber\\
    &iii)\, \mu_3=1,\, T_0=0.
\end{align}
\noindent In this case we have four different shifted nonlocal Maccari equations:
 \begin{align}
&1)\,au_t(x,y,t)+u_{xx}(x,y,t)+2\sigma_1D^{-1}_{y}[u(x,y,t)\bar{u}(x,y,t)+u(-x+x_0,-y+y_0,t)\nonumber \\  &\times \bar{u}(-x+x_0,-y+y_0,t)]_x u(x,y,t)=0,\,\, a=\bar{a},\label{B3.1}\\
&2)\, au_t(x,y,t)+u_{xx}(x,y,t)+\sigma_1D^{-1}_{y}[u(x,y,t)\bar{u}(x,y,t)+u(-x+x_0,-y+y_0,t)\bar{u}(-x+x_0,-y+y_0,t)\nonumber\\
&-u(-x+x_0,y,t)\bar{u}(-x+x_0,y,t)-u(x,-y+y_0,t)\bar{u}(x,-y+y_0,t)]_x  u(x,y,t)=0,\,\, a=\bar{a}, \label{B3.2}\\
&3)\, au_t(x,y,t)+u_{xx}(x,y,t)+\sigma_1D^{-1}_{y}[u(x,y,t)\bar{u}(x,y,t)+u(-x+x_0,-y+y_0,t)\nonumber \\  &\times \bar{u}(-x+x_0,-y+y_0,t)+u(x,y,-t+T_0)\bar{u}(x,y,-t+T_0)\nonumber \\  &+u(-x+x_0,-y+y_0,-t+T_0)\bar{u}(-x+x_0,-y+y_0,-t+T_0)]_xu(x,y,t)=0,\,\, a=-\bar{a}, \label{CRxy3}\\
&4)\,au_t(x,y,t)+u_{xx}(x,y,t)+\sigma_1D^{-1}_{y}[u(x,y,t)\bar{u}(x,y,t)+u(-x+x_0,-y+y_0,t)\nonumber \\  &\times \bar{u}(-x+x_0,-y+y_0,t)-u(x,-y+y_0,-t+T_0)\bar{u}(x,-y+y_0,-t+T_0)\nonumber \\  &-u(-x+x_0,y,-t+T_0)\bar{u}(-x+x_0,y,-t+T_0)]_xu(x,y,t)=0,\,\, a=-\bar{a}. \label{CRxy4}
\end{align}
The equation (\ref{B3.1}) is a two-place equation but (\ref{B3.2}), (\ref{CRxy3}), and (\ref{CRxy4}) are four-place equations.

\subsection{Shifted nonlocal Maccari equations obtained from the complex shifted nonlocal Maccari systems}

\noindent Similar to the real shifted nonlocal Maccari systems, we can further reduce the complex shifted nonlocal systems to shifted nonlocal equations by the reductions (\ref{rho3realtf}) and (\ref{rho3complextf}).

\noindent \textbf{Remark.} Recall that the real $x$-, $y$-, and $xy$-space shifted nonlocal Maccari systems and the complex $x$-, $y$-, and $xy$-space shifted nonlocal Maccari systems are the same if $a=\bar{a}$. Under this condition applying (\ref{rho3realtf}) gives exactly the same equations; (\ref{RRx1}) and (\ref{RRx2}) from the $x$-space shifted nonlocal systems, (\ref{RRy1}) and (\ref{RRy2}) from the $y$-space shifted nonlocal systems, (\ref{RRxy1}) and (\ref{RRxy2}) from the $xy$-space shifted nonlocal systems. Therefore, here we will not consider the equations obtained from complex $x$-, $y$-, and $xy$-space shifted nonlocal Maccari systems with the reduction (\ref{rho3realtf}). Similarly, if we apply (\ref{rho3complextf}) onto complex $x$-, $y$-, and $xy$-space shifted nonlocal Maccari systems we exactly get all the equations in B.I, B.II, and B.III.

\noindent \textbf{C.}\, We will first use the reduction (\ref{rho3realtf}) on the complex shifted nonlocal Maccari systems.

\noindent \textbf{C.I.}\, Applying the reduction (\ref{rho3realtf}) to the complex reverse time shifted nonlocal Maccari system \eqref{CMact1}-\eqref{CMact2}, where $a=-\bar{a}$, gives the below conditions to be satisfied:
\begin{equation}
\mu_3=1,\quad T_0=0,\quad \sigma_1=\mu_1\mu_2\sigma_3,    \label{A4con}
\end{equation} and
\begin{align}
    &i)\, \mu_1=1,\, X_0=0,\nonumber\\
    &ii)\, \mu_2=1,\, Y_0=0.
    \end{align}
\noindent Therefore, we obtain the following four different shifted nonlocal Maccari equations:
\begin{align}
&1)\, au_t(x,y,t)+u_{xx}(x,y,t)+2\sigma_1D^{-1}_{y}[u(x,y,t)\bar{u}(x,y,t)+u(x,y,-t+t_0)\bar{u}(x,y,-t+t_0)]_xu(x,y,t)=0, \label{RCt1}\\
&2)\,au_t(x,y,t)+u_{xx}(x,y,t)+\sigma_1D^{-1}_{y}[u(x,y,t)\bar{u}(x,y,t)+u(x,y,-t+t_0)\bar{u}(x,y,-t+t_0)\nonumber\\
&-u(x,-y+Y_0,t)\bar{u}(x,-y+Y_0,t)-u(x,-y+Y_0,-t+t_0)\bar{u}(x,-y+Y_0,-t+t_0)]_x\nonumber\\
&\times u(x,y,t)=0, \label{RCt2}\\
&3)\,au_t(x,y,t)+u_{xx}(x,y,t)+\sigma_1D^{-1}_{y}[u(x,y,t)\bar{u}(x,y,t)+u(x,y,-t+t_0)\bar{u}(x,y,-t+t_0)\nonumber \\  &-u(-x+X_0,y,t)\bar{u}(-x+X_0,y,t)-u(-x+X_0,y,-t+t_0)\bar{u}(-x+X_0,y,-t+t_0)]_xu(x,y,t)=0, \label{RCt3}\\
&4)\,au_t(x,y,t)+u_{xx}(x,y,t)+\sigma_1D^{-1}_{y}[u(x,y,t)\bar{u}(x,y,t)+u(x,y,-t+t_0)\bar{u}(x,y,-t+t_0)\nonumber \\  &+u(-x+X_0,-y+Y_0,t)\bar{u}(-x+X_0,-y+Y_0,t)+u(-x+X_0,-y+Y_0,-t+t_0)\nonumber\\
&\times \bar{u}(-x+X_0,-y+Y_0,-t+t_0)]_x u(x,y,t)=0. \label{RCt4}
\end{align}
The equation (\ref{RCt1}) is a two-place equation but the equations (\ref{RCt2}), (\ref{RCt3}), and (\ref{RCt4}) are four-place.

\noindent \textbf{C.II.}\, When we use the reduction \eqref{rho3realtf} on the complex reverse $x$-space-time shifted nonlocal Maccari system \eqref{CMacxt1}-\eqref{CMacxt2}, where $a=-\bar{a}$,  we have the conditions (\ref{A4con}) and
\begin{align}
    &i)\, \mu_1=-1,\, X_0=x_0\,;\, \mu_1=1,\, X_0=0,\nonumber\\
    &ii)\, \mu_2=1,\, Y_0=0,
    \end{align}
for consistency. We obtain four different shifted nonlocal Maccari equations;
\begin{align}
&1)\, au_t(x,y,t)+u_{xx}(x,y,t)+2\sigma_1D^{-1}_{y}[u(x,y,t)\bar{u}(x,y,t)-u(-x+x_0,y,-t+t_0)\bar{u}(-x+x_0,y,-t+t_0)]_x\nonumber \\  &\times u(x,y,t)=0, \label{RCxt1}\\
&2)\, au_t(x,y,t)+u_{xx}(x,y,t)+\sigma_1D^{-1}_{y}[u(x,y,t)\bar{u}(x,y,t)-u(-x+x_0,y,-t+t_0)\bar{u}(-x+x_0,y,-t+t_0)\nonumber \\  &-u(x,-y+Y_0,t)\bar{u}(x,-y+Y_0,t)+u(-x+x_0,-y+Y_0,-t+t_0)\bar{u}(-x+x_0,-y+Y_0,-t+t_0)]_x\nonumber \\  &\times u(x,y,t)=0, \label{RCxt2}\\
&3)\, au_t(x,y,t)+u_{xx}(x,y,t)+\sigma_1D^{-1}_{y}[u(x,y,t)\bar{u}(x,y,t)-u(-x+x_0,y,-t+t_0)\bar{u}(-x+x_0,y,-t+t_0)\nonumber \\  &-u(-x+x_0,y,t)\bar{u}(-x+x_0,y,t)+u(x,y,-t+t_0)\bar{u}(x,y,-t+t_0)]_x u(x,y,t)=0, \label{RCxt3}\\
&4)\, au_t(x,y,t)+u_{xx}(x,y,t)+\sigma_1D^{-1}_{y}[u(x,y,t)\bar{u}(x,y,t)-u(-x+x_0,y,-t+t_0) \bar{u}(-x+x_0,y,-t+t_0)\nonumber \\  &+u(-x+x_0,-y+Y_0,t)\bar{u}(-x+x_0,-y+Y_0,t)-u(x,-y+Y_0,-t+t_0)\bar{u}(x,-y+Y_0,-t+t_0)]_x
\nonumber \\  & \times u(x,y,t)=0, \label{RCxt4}
\end{align}
where only the first equation is a two-place equation but the others are four-place.

\noindent \textbf{C.III.}\, Applying (\ref{rho3realtf}) to the complex reverse $y$-space-time shifted nonlocal Maccari system \eqref{CMacyt1}-\eqref{CMacyt2}, where $a=-\bar{a}$, to derive shifted nonlocal Maccari equations consistently, we get the conditions (\ref{A4con}) and
\begin{align}
    &i)\, \mu_1=1,\, X_0=0,\nonumber\\
    &ii)\, \mu_2=-1,\, Y_0=y_0\,;\, \mu_2=1,\, Y_0=0.
 \end{align}
The resulting shifted nonlocal Maccari equations are
\begin{align}
&1)\, au_t(x,y,t)+u_{xx}(x,y,t)+2\sigma_1D^{-1}_{y}[u(x,y,t)\bar{u}(x,y,t)-u(x,-y+y_0,-t+t_0)\bar{u}(x,-y+y_0,-t+t_0)]_x
\nonumber \\  & \times u(x,y,t)=0, \label{RCyt1}\\
&2)\, au_t(x,y,t)+u_{xx}(x,y,t)+\sigma_1D^{-1}_{y}[u(x,y,t)\bar{u}(x,y,t)-u(x,-y+y_0,-t+t_0)\bar{u}(x,-y+y_0,-t+t_0)\nonumber \\  &-u(x,-y+y_0,t)\bar{u}(x,-y+y_0,t)+u(x,y,-t+t_0)\bar{u}(x,y,-t+t_0)]_x u(x,y,t)=0, \label{RCyt2}\\
&3)\, au_t(x,y,t)+u_{xx}(x,y,t)+\sigma_1D^{-1}_{y}[u(x,y,t)\bar{u}(x,y,t)-u(x,-y+y_0,-t+t_0)\bar{u}(x,-y+y_0,-t+t_0)
\nonumber \\  &-u(-x+X_0,y,t)\bar{u}(-x+X_0,y,t)+u(-x+X_0,-y+y_0,-t+t_0)\bar{u}(-x+X_0,-y+y_0,-t+t_0)]_x
\nonumber \\  &\times u(x,y,t)=0, \label{RCyt3}\\
&4)\, au_t(x,y,t)+u_{xx}(x,y,t)+\sigma_1D^{-1}_{y}[u(x,y,t)\bar{u}(x,y,t)-u(x,-y+y_0,-t+t_0) \bar{u}(x,-y+y_0,-t+t_0)
\nonumber \\  & +u(-x+X_0,-y+y_0,t)\bar{u}(-x+X_0,-y+y_0,t)-u(-x+X_0,y,-t+t_0)\bar{u}(-x+X_0,y,-t+t_0)]_x
\nonumber \\ &\times u(x,y,t)=0. \label{RCyt4}
\end{align}
The equation (\ref{RCyt1}) is two-place but (\ref{RCyt2}), (\ref{RCyt3}), and (\ref{RCyt4}) are four-place equations.

 \noindent \textbf{C.IV.}\, Using (\ref{rho3realtf}) on the complex reverse $xy$-space-time shifted nonlocal Maccari system \eqref{CMacxyt1}-\eqref{CMacxyt2}, where $a=-\bar{a}$, yields the conditions (\ref{A4con}) and
\begin{align}
    &i)\, \mu_1=-1,\, X_0=x_0\,;\, \mu_1=1,\, X_0=0,\nonumber\\
    &ii)\, \mu_2=-1,\, Y_0=y_0\,;\, \mu_2=1,\, Y_0=0.
    \end{align}
 Therefore, we get
\begin{align}
&1)\, au_t(x,y,t)+u_{xx}(x,y,t)+2\sigma_1D^{-1}_{y}[u(x,y,t)\bar{u}(x,y,t)+u(-x+x_0,-y+y_0,-t+t_0) \nonumber\\ &\times \bar{u}(-x+x_0,-y+y_0,-t+t_0)]_x
 u(x,y,t)=0, \label{RCxyt1}\\
&2)\, au_t(x,y,t)+u_{xx}(x,y,t)+\sigma_1D^{-1}_{y}[u(x,y,t)\bar{u}(x,y,t)+u(-x+x_0,-y+y_0,-t+t_0)\nonumber \\  &\times\bar{u}(-x+x_0,-y+y_0,-t+t_0)-u(x,-y+y_0,t)\bar{u}(x,-y+y_0,t)-u(-x+x_0,y,-t+t_0)\nonumber \\  &\times \bar{u}(-x+x_0,y,-t+t_0)]_x u(x,y,t)=0, \label{RCxyt2}\\
&3)\, au_t(x,y,t)+u_{xx}(x,y,t)+\sigma_1D^{-1}_{y}[u(x,y,t)\bar{u}(x,y,t)+u(-x+x_0,-y+y_0,-t+t_0)\nonumber \\  &\times\bar{u}(-x+x_0,-y+y_0,-t+t_0)-u(-x+x_0,y,t)\bar{u}(-x+x_0,y,t)-u(x,-y+y_0,-t+t_0)\nonumber \\  &\times\bar{u}(x,-y+y_0,-t+t_0)]_x u(x,y,t)=0, \label{RCxyt3}\\
&4)\, au_t(x,y,t)+u_{xx}(x,y,t)+\sigma_1D^{-1}_{y}[u(x,y,t)\bar{u}(x,y,t)+u(-x+x_0,-y+y_0,-t+t_0)\nonumber \\  &\times \bar{u}(-x+x_0,-y+y_0,-t+t_0)+u(-x+x_0,-y+y_0,t)\bar{u}(-x+x_0,-y+y_0,t)\nonumber \\  &+u(x,y,-t+t_0)\bar{u}(x,y,-t+t_0)]_x u(x,y,t)=0, \label{RCxyt4}
\end{align}
as consistent shifted nonlocal reductions of \eqref{CMacxyt1}-\eqref{CMacxyt2}. The first equation above is a two-place equation but the rest is four-place.

\noindent \textbf{D.}\, Now we use the reduction (\ref{rho3complextf}) on the complex shifted nonlocal  Maccari systems..

\noindent \textbf{D.I.}\, Application of the reduction \eqref{rho3complextf} on the complex reverse time shifted nonlocal Maccari system \eqref{CMact1}-\eqref{CMact2} gives the conditions
\begin{align}
 a=\bar{a}\mu_3,\quad \sigma_1=\mu_1\mu_2\sigma_3, \label{B4con}
\end{align} with
\begin{align}
    &i)\, \mu_1=1,\, X_0=0,\nonumber\\
    &ii)\, \mu_2=1,\, Y_0=0,\nonumber\\
    &iii)\, \mu_3=-1,\, T_0=t_0\,;\, \mu_3=1,\, T_0=0.
\end{align}
Since for the system \eqref{CMact1}-\eqref{CMact2} we have $a=-\bar{a}$, we will not consider the case $\mu_3=1, T_0=0$. The resulting shifted nonlocal Maccari equations are
\begin{align}
&1)\, au_t(x,y,t)+u_{xx}(x,y,t)+2\sigma_1D^{-1}_{y}[u(x,y,t)\bar{u}(x,y,t)+u(x,y,-t+t_0)\bar{u}(x,y,-t+t_0)]_x\nonumber \\  &\times u(x,y,t)=0, \label{CCt1}\\
&2)\, au_t(x,y,t)+u_{xx}(x,y,t)+\sigma_1D^{-1}_{y}[u(x,y,t)\bar{u}(x,y,t)+u(x,y,-t+t_0)\bar{u}(x,y,-t+t_0)\nonumber \\  &-u(x,-y+Y_0,t)\bar{u}(x,-y+Y_0,t)-u(x,-y+Y_0,-t+t_0)\bar{u}(x,-y+Y_0,-t+t_0)]_x\nonumber \\  &\times u(x,y,t)=0, \label{CCt2}\\
&3)\, au_t(x,y,t)+u_{xx}(x,y,t)+\sigma_1D^{-1}_{y}[u(x,y,t)\bar{u}(x,y,t)+u(x,y,-t+t_0)\bar{u}(x,y,-t+t_0)\nonumber \\  &-u(-x+X_0,y,t)\bar{u}(-x+X_0,y,t)-u(-x+X_0,y,-t+t_0)\bar{u}(-x+X_0,y,-t+t_0)]_x\nonumber \\  &\times u(x,y,t)=0, \label{CCt3}\\
&4)\, au_t(x,y,t)+u_{xx}(x,y,t)+\sigma_1D^{-1}_{y}[u(x,y,t)\bar{u}(x,y,t)+u(x,y,-t+t_0)\bar{u}(x,y,-t+t_0)\nonumber \\  &+u(-x+X_0,-y+Y_0,t)\bar{u}(-x+X_0,-y+Y_0,t)+u(-x+X_0,-y+Y_0,-t+t_0)\nonumber \\  &\times\bar{u}(-x+X_0,-y+Y_0,-t+t_0)]_x u(x,y,t)=0. \label{CCt4}
\end{align}
While the equation (\ref{CCt1}) is a two-place equation, the equations (\ref{CCt2}), (\ref{CCt3}), and (\ref{CCt4}) are four-place.

\noindent \textbf{D.II.}\, Using the reduction \eqref{rho3complextf} on the complex reverse $x$-space-time shifted nonlocal Maccari system \eqref{CMacxt1}-\eqref{CMacxt2} requires the conditions (\ref{B4con}) with
\begin{align}
    &i)\, \mu_1=-1,\, X_0=x_0\,;\, \mu_1=1,\, X_0=0,\nonumber\\
    &ii)\, \mu_2=1,\, Y_0=0,\nonumber\\
    &iii)\, \mu_3=-1,\, T_0=t_0\,;\, \mu_3=1,\, T_0=0,
\end{align}
to be satisfied for consistency. Note that for the system \eqref{CMacxt1}-\eqref{CMacxt2} we already have $a=-\bar{a}$, so the case $\mu_3=1, T_0=0$ is not possible here. We obtain four different shifted nonlocal Maccari equations:
\begin{align}
&1)\, au_t(x,y,t)+u_{xx}(x,y,t)+2\sigma_1D^{-1}_{y}[u(x,y,t)\bar{u}(x,y,t)-u(-x+x_0,y,-t+t_0)\bar{u}(-x+x_0,y,-t+t_0)]_x\nonumber \\  &\times u(x,y,t)=0, \label{CCxt1}\\
&2)\, au_t(x,y,t)+u_{xx}(x,y,t)+\sigma_1D^{-1}_{y}[u(x,y,t)\bar{u}(x,y,t)-u(-x+x_0,y,-t+t_0)\bar{u}(-x+x_0,y,-t+t_0)\nonumber \\  &-u(x,-y+Y_0,t)\bar{u}(x,-y+Y_0,t)+u(-x+x_0,-y+Y_0,-t+t_0)\bar{u}(-x+x_0,-y+Y_0,-t+t_0)]_x\nonumber \\  &\times u(x,y,t)=0, \label{CCxt2}\\
&3)\, au_t(x,y,t)+u_{xx}(x,y,t)+\sigma_1D^{-1}_{y}[u(x,y,t)\bar{u}(x,y,t)-u(-x+x_0,y,-t+t_0)\bar{u}(-x+x_0,y,-t+t_0)\nonumber \\  &-u(-x+x_0,y,t)\bar{u}(-x+x_0,y,t)+u(x,y,-t+t_0)\bar{u}(x,y,-t+t_0)]_xu(x,y,t)=0, \label{CCxt3}\\
&4)\, au_t(x,y,t)+u_{xx}(x,y,t)+\sigma_1D^{-1}_{y}[u(x,y,t)\bar{u}(x,y,t)-u(-x+x_0,y,-t+t_0) \bar{u}(-x+x_0,y,-t+t_0)
\nonumber \\& +u(-x+x_0,-y+Y_0,t)\bar{u}(-x+x_0,-y+Y_0,t)-u(x,-y+Y_0,-t+t_0)\bar{u}(x,-y+Y_0,-t+t_0)]_x\nonumber \\  &\times u(x,y,t)=0. \label{CCxt4}
\end{align}
Here the first equation is a two-place equation, but (\ref{CCxt2}), (\ref{CCxt3}), and (\ref{CCxt4}) are four-place.

\noindent \textbf{D.III.}\, Now applying the reduction \eqref{rho3complextf} to the complex reverse $y$-space-time shifted nonlocal Maccari system \eqref{CMacyt1}-\eqref{CMacyt2} we obtain the conditions (\ref{B4con}) with
\begin{align}
    &i)\, \mu_1=1,\, X_0=0,\nonumber\\
    &ii)\, \mu_2=-1,\, Y_0=y_0\,;\, \mu_2=1,\, Y_0=0,\nonumber\\
    &iii)\, \mu_3=-1,\, T_0=t_0\,;\, \mu_3=1,\, T_0=0,
\end{align}
to get shifted nonlocal equations consistently. Recall that for the system \eqref{CMacyt1}-\eqref{CMacyt2} we already have $a=-\bar{a}$, therefore, the case $\mu_3=1, T_0=0$ is not possible here. The resulting reduced equations are
\begin{align}
&1)\, au_t(x,y,t)+u_{xx}(x,y,t)+2\sigma_1D^{-1}_{y}[u(x,y,t)\bar{u}(x,y,t)-u(x,-y+y_0,-t+t_0)\bar{u}(x,-y+y_0,-t+t_0)]_x\nonumber \\  &\times u(x,y,t)=0, \label{CCyt1}\\
&2)\, au_t(x,y,t)+u_{xx}(x,y,t)+\sigma_1D^{-1}_{y}[u(x,y,t)\bar{u}(x,y,t)-u(x,-y+y_0,-t+t_0)\bar{u}(x,-y+y_0,-t+t_0)\nonumber \\  &-u(x,-y+y_0,t)\bar{u}(x,-y+y_0,t)+u(x,y,-t+t_0)\bar{u}(x,y,-t+t_0)]_x u(x,y,t)=0, \label{CCyt2}\\
&3)\, au_t(x,y,t)+u_{xx}(x,y,t)+\sigma_1D^{-1}_{y}[u(x,y,t)\bar{u}(x,y,t)-u(x,-y+y_0,-t+t_0)\bar{u}(x,-y+y_0,-t+t_0)\nonumber \\  &-u(-x+X_0,y,t)\bar{u}(-x+X_0,y,t)+u(-x+X_0,-y+y_0,-t+t_0)\bar{u}(-x+X_0,-y+y_0,-t+t_0)]_x\nonumber \\  &\times u(x,y,t)=0, \label{CCyt3}\\
&4)\, au_t(x,y,t)+u_{xx}(x,y,t)+\sigma_1D^{-1}_{y}[u(x,y,t)\bar{u}(x,y,t)-u(x,-y+y_0,-t+t_0)\bar{u}(x,-y+y_0,-t+t_0)
\nonumber\\&+u(-x+X_0,-y+y_0,t)\bar{u}(-x+X_0,-y+y_0,t)-u(-x+X_0,y,-t+t_0)\bar{u}(-x+X_0,y,-t+t_0)]_x\nonumber \\  &\times u(x,y,t)=0. \label{CCyt4}
\end{align}
The equation (\ref{CCyt1}) is two-place while (\ref{CCyt2}), (\ref{CCyt3}), and (\ref{CCyt4}) are four-place equations.

\noindent \textbf{D.IV.}\, Finally, applying the reduction \eqref{rho3complextf} to the complex reverse $xy$-space-time shifted nonlocal Maccari system \eqref{CMacxyt1}-\eqref{CMacxyt2} gives the conditions (\ref{B4con}) with
\begin{align}
    &i)\, \mu_1=-1,\, X_0=x_0\,;\, \mu_1=1,\, X_0=0,\nonumber\\
    &ii)\, \mu_2=-1,\, Y_0=y_0\,;\, \mu_2=1,\, Y_0=0,\nonumber\\
    &iii)\, \mu_3=-1,\, T_0=t_0\,;\, \mu_3=1,\, T_0=0,
\end{align}
yielding the following four different shifted nonlocal Maccari equations:
\begin{align}
&1)\, au_t(x,y,t)+u_{xx}(x,y,t)+2\sigma_1D^{-1}_{y}[u(x,y,t)\bar{u}(x,y,t)+u(-x+x_0,-y+y_0,-t+t_0)\nonumber \\  &\times\bar{u}(-x+x_0,-y+y_0,-t+t_0)]_x u(x,y,t)=0, \label{CCxyt1}\\
&2)\, au_t(x,y,t)+u_{xx}(x,y,t)+\sigma_1D^{-1}_{y}[u(x,y,t)\bar{u}(x,y,t)+u(-x+x_0,-y+y_0,-t+t_0)
\nonumber \\  &\times\bar{u}(-x+x_0,-y+y_0,-t+t_0)-u(x,-y+y_0,t)\bar{u}(x,-y+y_0,t)-u(-x+x_0,y,-t+t_0)\nonumber \\  &\times \bar{u}(-x+x_0,y,-t+t_0)]_x u(x,y,t)=0, \label{CCxyt2}\\
&3)\, au_t(x,y,t)+u_{xx}(x,y,t)+\sigma_1D^{-1}_{y}[u(x,y,t)\bar{u}(x,y,t)+u(-x+x_0,-y+y_0,-t+t_0)\nonumber \\  &\times\bar{u}(-x+x_0,-y+y_0,-t+t_0)-u(-x+x_0,y,t)\bar{u}(-x+x_0,y,t)-u(x,-y+y_0,-t+t_0)\nonumber \\  &\times\bar{u}(x,-y+y_0,-t+t_0)]_x u(x,y,t)=0, \label{CCxyt3}\\
&4)\, au_t(x,y,t)+u_{xx}(x,y,t)+\sigma_1D^{-1}_{y}[u(x,y,t)\bar{u}(x,y,t)+u(-x+x_0,-y+y_0,-t+t_0)\nonumber \\  &\times \bar{u}(-x+x_0,-y+y_0,-t+t_0)+u(-x+x_0,-y+y_0,t)\bar{u}(-x+x_0,-y+y_0,t)\nonumber \\  &+u(x,y,-t+t_0)\bar{u}(x,y,-t+t_0)]_x u(x,y,t)=0. \label{CCxyt4}
\end{align}
In the above list of equations, while the equation (\ref{CCxyt1}) is two-place, the equations (\ref{CCxyt2}), (\ref{CCxyt3}), and (\ref{CCxyt4}) are four-place. Note that similar to the previous cases the system \eqref{CMacxyt1}-\eqref{CMacxyt2} requires $a=-\bar{a}$ so the case $\mu_3=1, T_0=0$ is not valid here.

\section{Soliton solutions of the $5$-component Maccari system}

To obtain soliton solutions of the $5$-component Maccari system \eqref{M1}-\eqref{M5} we use the Hirota method.

\noindent \textbf{Step 1: Bilinearization}. Consider the following rational transformation as bilinearizing transformation for the Maccari system \eqref{M1}-\eqref{M5},
\begin{align}
u=\frac{g}{f}, \quad v=\frac{h}{f}, \quad w=\frac{s}{f}, \quad z=\frac{q}{f}, \quad p=2(\ln f)_{xx}.  \label{TFMac}
\end{align}
Equations \eqref{M1}-\eqref{M4}, respectively, become
\begin{align}
&ag_tf-agf_t+g_{xx}-2g_xf_x+gf_{xx}=0,\label{eq1} \\
&ah_tf-ahf_t+h_{xx}-2h_xf_x+hf_{xx}=0,\label{eq2} \\
&as_tf-asf_t+s_{xx}-2s_xf_x+sf_{xx}=0,\label{eq3} \\
&aq_tf-aqf_t+q_{xx}-2q_xf_x+qf_{xx}=0,\label{eq4}
\end{align}
and \eqref{M5} turns to be
\begin{align}
&\sigma_1\bigg(-\frac{\bar{g}g_x}{f^2}-\frac{g\bar{g}_x}{f^2}+\frac{2g\bar{g}f_x}{f^3}\bigg)
+\sigma_2\bigg(-\frac{\bar{h}h_x}{f^2}-\frac{h\bar{h}_x}{f^2}+\frac{2h\bar{h}f_x}{f^3}\bigg)+\sigma_3
\bigg(-\frac{\bar{s}s_x}{f^2}-\frac{s\bar{s}_x}{f^2}+\frac{2s\bar{s}f_x}{f^3}\bigg) \nonumber\\
&+\sigma_4\bigg(-\frac{\bar{q}q_x}{f^2}-\frac{q\bar{q}_x}{f^2}+\frac{2q\bar{q}f_x}{f^3}\bigg)+\frac{2f_{xxy}f}{f^2}-\frac{2f_{xx}f_y}{f^2}-\frac{4f_xf_{xy}}{f^2}
+\frac{4f^2_xf_y}{f^3}=0.\label{eq5}
\end{align}

\noindent \textbf{Step 2: Transformation into Hirota bilinear form}. Here we use the Hirota $D$-operator which is a special differential operator given by
\begin{equation}
D_t^nD_x^m\{F\cdot G\}=\Big(\frac{\partial}{\partial t}-\frac{\partial}{\partial t'}\Big)^n\Big(\frac{\partial}{\partial x}-\frac{\partial}{\partial x'}\Big)^m\,F(x,y,t)G(x',y',t')|_{x'=x,t'=t,y'=y},\quad m, n \in \mathbb{N}.
\end{equation}
Arranging the equations (\ref{eq1})-(\ref{eq5}) by using the Hirota $D$-operator we get the Hirota bilinear
form of the $5$-component Maccari system as
\begin{align}
P_1(D)\{g.f\}&=(aD_t+D^2_x)\{g.f\}=0, \label{P1Mac}\\
P_2(D)\{h.f\}&=(aD_t+D^2_x)\{h.f\}=0, \label{P2Mac}\\
P_3(D)\{s.f\}&=(aD_t+D^2_x)\{s.f\}=0, \label{P3Mac}\\
P_4(D)\{q.f\}&=(aD_t+D^2_x)\{q.f\}=0, \label{P4Mac}\\
D_xD_y\{f.f\}&=\sigma_1g\bar{g}+\sigma_2h\bar{h}+\sigma_3s\bar{s}+\sigma_4q\bar{q}. \label{P5Mac}
\end{align}

\noindent \textbf{Step 3: Hirota perturbation}. To find $N$-soliton solutions of the $5$-component Maccari system, we use the following expansions for the functions $g$, $h$, $s$, $q$, and $f$:
\begin{align}
g=\sum_{j=1}^{N}\varepsilon^{2j-1}g_{2j-1}, \quad h=\sum_{j=1}^{N}\varepsilon^{2j-1}h_{2j-1}, \quad
s=\sum_{j=1}^{N}\varepsilon^{2j-1}s_{2j-1},\quad q=\sum_{j=1}^{N}\varepsilon^{2j-1}q_{2j-1}, \quad f=\sum_{j=0}^{N}\varepsilon^{2j}f_{2j}.
\end{align}
For one-soliton solution, we take $N=1$, i.e.,
\begin{align}
g=\varepsilon g_1, \quad h=\varepsilon h_1, \quad s=\varepsilon s_1, \quad q=\varepsilon q_1, \quad f=1+\varepsilon^2 f_2,\label{expMac}
\end{align}
where
\begin{align}
g_1=e^{\theta_1}, \quad h_1=e^{\theta_2}, \quad s_1=e^{\theta_3}, \quad q_1=e^{\theta_4}, \quad \theta_{i}=k_ix+l_iy+\omega_it+\alpha_i,\quad i=1,2,3,4.
\end{align}
Here $k_i, l_i, \omega_i, \alpha_i, \, i=1,2,3,4$ are constants. We substitute \eqref{expMac} into \eqref{P1Mac}-\eqref{P5Mac} and make the coefficients of $\varepsilon^n$, $n=1,2,3,4$, equal to zero. The coefficients of $\varepsilon^1$ yield the dispersion relations
\begin{equation}\label{dispersion}
\omega_i=-\frac{k_i^2}{a},\quad i=1,2,3,4.
\end{equation}
From the coefficient of $\varepsilon^2$ we have $f_2$ as
\begin{equation}\displaystyle
f_2=\sum_{i=1}^4 \frac{A_i}{2}e^{\theta_i+\bar{\theta}_i},\quad A_i=\frac{\sigma_i}{(k_i+\bar{k}_i)(l_i+\bar{l}_i)},\quad i=1,2,3,4.
\end{equation}
From the coefficients of $\varepsilon^3$ we have
\begin{align*}
&\bigg(1+\frac{a}{\bar{a}}\bigg)(A_1(\bar{k}_1)^2e^{\theta_1+\bar{\theta}_1})-2[A_2(k_1-k_2)(k_2+\bar{k}_2)e^{\theta_2+\bar{\theta}_2}+A_3(k_1-k_3)(k_3+\bar{k}_3)e^{\theta_3+\bar{\theta}_3}\nonumber \\
&+A_4(k_1-k_4)(k_4+\bar{k}_4)e^{\theta_4+\bar{\theta}_4}]+\bigg(1+\frac{a}{\bar{a}}\bigg)[A_2(\bar{k}_2)^2e^{\theta_2+\bar{\theta}_2}+A_3(\bar{k}_3)^2e^{\theta_3+\bar{\theta}_3}+A_4(\bar{k}_4)^2e^{\theta_4+\bar{\theta}_4}]=0.
\end{align*}
To satisfy the above equation, we must have $\bar{a}=-a$ and $k_1=k_2=k_3=k_4$. The coefficients of $\varepsilon^4$ vanish directly under these conditions. Hence we obtain the one-soliton solution of the Maccari system \eqref{M1}-\eqref{M5} as
\begin{align}
&u=\frac{e^{k_1x+l_1y+\omega_1t+\alpha_1}}{1+\frac{e^{(k_1+\bar{k}_1)x+(\omega_1+\bar{\omega}_1)t}}{2(k_1+\bar{k}_1)}\Big(\sum\limits_{i=1}^4 \frac{\sigma_i e^{(l_i+\bar{l}_i)y+\alpha_i+\bar{\alpha}_i}}{(l_i+\bar{l}_i)}\Big)},\,\,\, v=\frac{e^{k_1x+l_2y+\omega_1t+\alpha_2}}{1+\frac{e^{(k_1+\bar{k}_1)x+(\omega_1+\bar{\omega}_1)t}}{2(k_1+\bar{k}_1)}\Big(\sum\limits_{i=1}^4 \frac{\sigma_i e^{(l_i+\bar{l}_i)y+\alpha_i+\bar{\alpha}_i}}{(l_i+\bar{l}_i)}\Big)},\nonumber\\
&w=\frac{e^{k_1x+l_3y+\omega_1t+\alpha_3}}{1+\frac{e^{(k_1+\bar{k}_1)x+(\omega_1+\bar{\omega}_1)t}}{2(k_1+\bar{k}_1)}\Big(\sum\limits_{i=1}^4 \frac{\sigma_i e^{(l_i+\bar{l}_i)y+\alpha_i+\bar{\alpha}_i}}{(l_i+\bar{l}_i)}\Big)}, \,\,\, z=\frac{e^{k_1x+l_4y+\omega_1t+\alpha_4}}{1+\frac{e^{(k_1+\bar{k}_1)x+(\omega_1+\bar{\omega}_1)t}}{2(k_1+\bar{k}_1)}\Big(\sum\limits_{i=1}^4 \frac{\sigma_i e^{(l_i+\bar{l}_i)y+\alpha_i+\bar{\alpha}_i}}{(l_i+\bar{l}_i)}\Big)},\nonumber\\
&\hspace{3.5cm} p=\frac{(k_1+\bar{k}_1)e^{(k_1+\bar{k}_1)x+(\omega_1+\bar{\omega}_1)t}\Big(\sum\limits_{i=1}^4 \frac{\sigma_i e^{(l_i+\bar{l}_i)y+\alpha_i+\bar{\alpha}_i}}{(l_i+\bar{l}_i)}\Big)}{\Big[1+\frac{e^{(k_1+\bar{k}_1)x+(\omega_1+\bar{\omega}_1)t}}{2(k_1+\bar{k}_1)}
\Big(\sum\limits_{i=1}^4 \frac{\sigma_i e^{(l_i+\bar{l}_i)y+\alpha_i+\bar{\alpha}_i}}{(l_i+\bar{l}_i)}\Big)\Big]^2},\nonumber\\
\label{sol-Maccari}
\end{align}
with $\theta_i=k_ix+l_iy+\omega_it+\alpha_i,\, \omega_i=-\frac{k_i^2}{a}, \, i=1,2,3,4$. Here $k_i, l_i, \alpha_i$ are any complex numbers and $a$ is a pure imaginary number.

\section{Soliton solutions of the reduced shifted nonlocal Maccari systems}

In this section we shall use Type 1 and Type 2 approaches \cite{gur1}, \cite{gur3} with the reduction formulas to obtain one-soliton solutions of the reduced shifted nonlocal Maccari systems.

\noindent \textbf{i)}\, $v(x,y,t)=\rho_1 u(\varepsilon_1x+x_0,\varepsilon_2y+y_0,t)$, $z(x,y,t)=\rho_2 \omega(\varepsilon_1x+x_0,\varepsilon_2y+y_0,t)$, $\rho_1^2=\rho_2^2=1$, $x_0,y_0 \in \mathbb{R}$.

  \noindent Let us use Type 1 with this reduction and one-soliton solution (\ref{sol-Maccari}) of the $5$-component Maccari system. We obtain
	\begin{equation}
	k_1=\varepsilon_1 k_1, \quad l_2=\varepsilon_2 l_1, \quad l_4=\varepsilon_2 l_3,\quad e^{\alpha_4}=\rho_1 e^{k_1x_0+l_3y_0+\alpha_3}, \quad e^{\alpha_2}=\rho_1 e^{k_1x_0+l_1y_0+\alpha_1},
    \end{equation}
with
\begin{equation}
 \sigma_2=\varepsilon_1 \varepsilon_2 \sigma_1, \quad \sigma_4=\varepsilon_1 \varepsilon_2 \sigma_3.
 \end{equation}
 It follows that $\varepsilon_1=1$, $x_0=0$, and $\varepsilon _ 2=\pm 1$. For having nonlocal reduction we only have the case $\varepsilon_2=-1$, that is, the case where $(\varepsilon_1,\varepsilon_2,\varepsilon_3)=(1,-1,1)$. The constraints become
\begin{align}
l_2=-l_1,\quad l_4=-l_3,\quad \omega_2=\omega_1,\quad \sigma_2=-\sigma_1, \quad \sigma_4=-\sigma_3.
\end{align}
 Hence one-soliton solution of the real $y$-space shifted nonlocal Maccari system \eqref{RMacy1}-\eqref{RMacy2} is
\begin{equation}
u=\frac{e^{k_1x+l_1y+\omega_1t+\alpha_1}}{Q},\quad w=\frac{e^{k_1x+l_3y+\omega_1t+\alpha_3}}{Q}, \quad p=\frac{A(k_1+\bar{k}_1)e^{(k_1+\bar{k}_1)x+(\omega_1+\bar{\omega}_1)t}}{\Big[1+A\frac{e^{(k_1+\bar{k}_1)x+(\omega_1+\bar{\omega}_1)t}}{2(k_1+\bar{k}_1)}\Big]^2},\label{sol-realy}
\end{equation}
where
{\small\begin{align*}
&Q=1+\frac{e^{(k_1+\bar{k}_1)x+(\omega_1+\bar{\omega}_1)t}}{2(k_1+\bar{k}_1)}
\Big[\sigma_1\frac{e^{\alpha_1+\bar{\alpha}_1}}{(l_1+\bar{l}_1)}\Big(e^{(l_1+\bar{l}_1)y}+e^{(l_1+\bar{l}_1)(-y+y_0)}\Big)+\sigma_3\frac{e^{\alpha_3+\bar{\alpha}_3}}{(l_3+\bar{l}_3)}\Big(e^{(l_3+\bar{l}_3)y}+e^{(l_3+\bar{l}_3)(-y+y_0)}\Big) \Big],\\
&A=\Big[\sigma_1\frac{e^{\alpha_1+\bar{\alpha}_1}}{(l_1+\bar{l}_1)}\Big( e^{(l_1+\bar{l}_1)y}+e^{(l_1+\bar{l}_1)(-y+y_0)}\Big)+\sigma_3\frac{e^{\alpha_3+\bar{\alpha}_3}}{(l_3+\bar{l}_3)}\Big( e^{(l_3+\bar{l}_3)y}+e^{(l_3+\bar{l}_3)(-y+y_0)}\Big)\Big].
\end{align*}}

We now use Type 2 with this reduction and solution (\ref{sol-Maccari}) to obtain solutions for other cases. In this case we have the following constraints for $(\varepsilon_1,\varepsilon_2,\varepsilon_3)=(-1,-1,1)$  as
\begin{align} \label{xycon}
&k_1=\bar{k}_1, \quad \omega_1=-\bar{\omega}_1, \quad l_2=\bar{l}_1,\quad l_4=\bar{l}_3,\quad l_3+\bar{l}_3=l_1+\bar{l}_1, \nonumber \\
&\sigma_2=\sigma_1, \quad \sigma_4=\sigma_3, \quad e^{2\alpha_2}=e^{2\alpha_1+(l_1-\bar{l}_1)y_0},\quad e^{2\alpha_4}=e^{2\alpha_3+(l_3-\bar{l}_3)y_0}.
\end{align}
Hence one-soliton solution of the real $xy$-space shifted nonlocal Maccari system \eqref{RMacxy1}-\eqref{RMacxy2} is
\begin{align}
&u=\frac{e^{k_1x+l_1y+\omega_1t+\alpha_1}}{1+\frac{e^{2k_1x+(l_1+\bar{l}_1)y}}{2k_1(l_1+\bar{l}_1)}\Big(\sigma_1
e^{\alpha_1+\bar{\alpha}_1}+\sigma_3e^{\alpha_3+\bar{\alpha}_3}\Big)}, \quad
w=\frac{e^{k_1x+l_3y+\omega_1t+\alpha_3}}{1+\frac{e^{2k_1x+(l_1+\bar{l}_1)y}}{2k_1(l_1+\bar{l}_1)}
\Big(\sigma_1e^{\alpha_1+\bar{\alpha}_1}+\sigma_3e^{\alpha_3+\bar{\alpha}_3}\Big)},\nonumber\\
&\hspace{3.5cm} p=\frac{\frac{4k_1}{(l_1+\bar{l}_1)}e^{2k_1x+(l_1+\bar{l}_1)y}\Big(\sigma_1e^{\alpha_1+\bar{\alpha}_1}
+\sigma_3e^{\alpha_3+\bar{\alpha}_3}\Big)}{\Big[1+\frac{e^{2k_1x+(l_1+\bar{l}_1)y}}{2k_1(l_1+\bar{l}_1)}
\Big(\sigma_1e^{\alpha_1+\bar{\alpha}_1}+\sigma_3e^{\alpha_3+\bar{\alpha}_3}\Big)\Big]^2}.\nonumber\\
&\label{sol-realxy}
\end{align}

Note that nontrivial solution for the case $(\varepsilon_1,\varepsilon_2,\varepsilon_3)=(-1,1,1)$ cannot be derived with Type 1 or Type 2 approaches. Different solution methods can be used to obtain nontrivial solutions for this case.
\medskip

\noindent \textbf{ii)}\, $v(x,y,t)=\rho_1 \bar{u}(\varepsilon_1x+x_0,\varepsilon_2y+y_0,\varepsilon_3t+t_0)$, $z(x,y,t)=\rho_2 \bar{\omega}(\varepsilon_1x+x_0,\varepsilon_2y+y_0,\varepsilon_3t+t_0)$, $\rho_1^2=\rho_2^2=1$, $x_0,y_0,t_0 \in \mathbb{R}$.

  We first use Type 1 with this reduction and one-soliton solution (\ref{sol-Maccari}). We obtain
	\begin{equation}
	k_1=\varepsilon_1 \bar{k}_1,\,\, l_2=\varepsilon_2 \bar{l}_1,\,\, l_4=\varepsilon_2 \bar{l}_3,\,\, \omega_1=\varepsilon_3 \bar{\omega}_1,\,\, e^{\alpha_2}=\rho_1 e^{\bar{k}_1x_0+\bar{l}_1y_0+\bar{\omega}_1t_0+\bar{\alpha}_1},\,\, e^{\alpha_4}=\rho_1 e^{\bar{k}_1x_0+\bar{l}_3y_0+\bar{\omega}_1t_0+\bar{\alpha}_3}
	\end{equation}
with
 \begin{equation}
          \sigma_2=\varepsilon_1 \varepsilon_2 \sigma_1, \quad \sigma_4=\varepsilon_1 \varepsilon_2 \sigma_3.
 \end{equation}
Here we choose $\varepsilon_1=1$ to avoid trivial solution. Thus we get the following possibilities:
 \begin{align}
    &a)\, \varepsilon_1=1, \varepsilon_2=1, \varepsilon_3=-1,\, x_0=y_0=0,\\
    &b)\, \varepsilon_1=1, \varepsilon_2=-1, \varepsilon_3=1,\, x_0=t_0=0,\\
    &c)\, \varepsilon_1=1, \varepsilon_2=-1, \varepsilon_3=-1,\, x_0=0,
\end{align}
corresponding to the cases $(i)$, $(ii)$, and $(iii)$ in \eqref{CMacpos}. Note that the case b) gives $a=\bar{a}$ but we have previously obtained $a=-\bar{a}$ for the applicability of the Hirota method. Therefore, we cannot find solutions for the case $(\varepsilon_1,\varepsilon_2,\varepsilon_3)=(1,-1,1)$ using the Hirota method.

The constraints for the case  $(\varepsilon_1,\varepsilon_2,\varepsilon_3)=(1,1,-1)$ are given as
\begin{align}
l_2=\bar{l}_1,\,\, l_4=\bar{l}_3,\,\, k_1=\bar{k}_1, \,\, \omega_1=-\bar{\omega}_1,\,\, \sigma_2=\sigma_1,\,\,  \sigma_4=\sigma_3,\,\,
e^{\alpha_2}=\rho_1 e^{\bar{\omega}_1t_0+\bar{\alpha}_1},\,\, e^{\alpha_4}=\rho_1 e^{\bar{\omega}_1t_0+\bar{\alpha}_3}.
\end{align}
Hence one-soliton solution of the complex reverse time shifted nonlocal Maccari system \eqref{CMact1}-\eqref{CMact2} is
{\small\begin{align}
&u=\frac{e^{k_1x+l_1y+\omega_1t+\alpha_1}}{1+\frac{e^{2k_1x}}{2k_1}\Big(\sigma_1\frac{e^{(l_1+\bar{l}_1)y+\alpha_1
+\bar{\alpha}_1}}{(l_1+\bar{l}_1)}+\sigma_3\frac{e^{(l_3+\bar{l}_3)y+\alpha_3+\bar{\alpha}_3}}{(l_3+\bar{l}_3)}\Big)}, \,
w=\frac{e^{k_1x+l_3y+\omega_1t+\alpha_3}}{1+\frac{e^{2k_1x}}{2k_1}\Big(\sigma_1\frac{e^{(l_1+\bar{l}_1)y
+\alpha_1+\bar{\alpha}_1}}{(l_1+\bar{l}_1)}+\sigma_3\frac{e^{(l_3+\bar{l}_3)y+\alpha_3+\bar{\alpha}_3}}{(l_3+\bar{l}_3)}\Big)},\nonumber\\
&\hspace{3.5cm} p=\frac{2k_1e^{2k_1x}\Big(\sigma_1\frac{ e^{(l_1+\bar{l}_1)y+\alpha_1+\bar{\alpha}_1}}{(l_1+\bar{l}_1)}+\sigma_3\frac{ e^{(l_3+\bar{l}_3)y+\alpha_3+\bar{\alpha}_3}}{(l_3+\bar{l}_3)}\Big)}{\Big[1+\frac{e^{2k_1x}}{2k_1}\Big(\sigma_1
\frac{e^{(l_1+\bar{l}_1)y+\alpha_1+\bar{\alpha}_1}}{(l_1+\bar{l}_1)}+\sigma_3\frac{e^{(l_3+\bar{l}_3)y+\alpha_3+\bar{\alpha}_3}}{(l_3+\bar{l}_3)}\Big)\Big]^2}.\nonumber\\
&\label{sol-compt}
\end{align}}
\smallskip

We get the following constraints for the case $(\varepsilon_1,\varepsilon_2,\varepsilon_3)=(1,-1,-1)$:
\begin{align}
l_2=-\bar{l}_1,\,\, l_4=-\bar{l}_3,\,\, k_1=\bar{k}_1,\,\, \omega_1=-\bar{\omega}_1,\,\,  e^{\alpha_2}=\rho_1 e^{\bar{l}_1y_0+\bar{\omega}_1t_0+\bar{\alpha}_1},\,\, e^{\alpha_4}=\rho_1 e^{\bar{l}_3y_0+\bar{\omega}_1t_0+\bar{\alpha}_3},
\end{align}
with $\sigma_2=-\sigma_1,\, \sigma_4=-\sigma_3$. Therefore, we obtain one-soliton solution of the complex reverse $y$-space-time shifted nonlocal Maccari system \eqref{CMacyt1}-\eqref{CMacyt2} as
\begin{align}
&u =\frac{e^{k_1x+l_1y+\omega_1t+\alpha_1}}{1+\frac{e^{2k_1x}}{4k_1}\Big[\sigma_1\frac{e^{\alpha_1+\bar{\alpha}_1}}{(l_1+\bar{l}_1)}\Big( e^{(l_1+\bar{l}_1)y}+e^{(l_1+\bar{l}_1)(-y+y_0)}\Big)+\sigma_3\frac{e^{\alpha_3+\bar{\alpha}_3}}{(l_3+\bar{l}_3)}\Big( e^{(l_3+\bar{l}_3)y}+e^{(l_3+\bar{l}_3)(-y+y_0)}\Big)\Big]}, \nonumber\\
&w=\frac{e^{k_1x+l_3y+\omega_1t+\alpha_3}}{1+\frac{e^{2k_1x}}{4k_1}\Big[\sigma_1\frac{e^{\alpha_1+\bar{\alpha}_1}}{(l_1+\bar{l}_1)}\Big( e^{(l_1+\bar{l}_1)y}+e^{(l_1+\bar{l}_1)(-y+y_0)}\Big) +\sigma_3\frac{e^{\alpha_3+\bar{\alpha}_3}}{(l_3+\bar{l}_3)}\Big( e^{(l_3+\bar{l}_3)y}+e^{(l_3+\bar{l}_3)(-y+y_0)}\Big)\Big]}, \nonumber\\
&p=\frac{2k_1e^{2k_1x}\Big[\sigma_1\frac{e^{\alpha_1+\bar{\alpha}_1}}{(l_1+\bar{l}_1)}\Big( e^{(l_1+\bar{l}_1)y}+e^{(l_1+\bar{l}_1)(-y+y_0)}\Big)+\sigma_3\frac{e^{\alpha_3+\bar{\alpha}_3}}{(l_3+\bar{l}_3)}\Big( e^{(l_3+\bar{l}_3)y}+e^{(l_3+\bar{l}_3)(-y+y_0)}\Big)\Big]}{\Big[1+\frac{e^{2k_1x}}{4k_1}\Big(\sigma_1\frac{e^{\alpha_1+\bar{\alpha}_1}}{(l_1+\bar{l}_1)}\Big( e^{(l_1+\bar{l}_1)y}+e^{(l_1+\bar{l}_1)(-y+y_0)}\Big) +\sigma_3\frac{e^{\alpha_3+\bar{\alpha}_3}}{(l_3+\bar{l}_3)}\Big( e^{(l_3+\bar{l}_3)y}+e^{(l_3+\bar{l}_3)(-y+y_0)}\Big)\Big)\Big]^2}.\nonumber\\
&\label{sol-compyt}
\end{align}

Use now Type 2 approach. We obtain the constraints for the case $(\varepsilon_1,\varepsilon_2,\varepsilon_3)=(-1,-1,-1)$ as
\begin{align}
&l_2=l_1,\quad l_4=l_3,\quad l_3+\bar{l}_3=l_1+\bar{l}_1,\quad \sigma_2=\sigma_1, \quad \sigma_4=\sigma_3, \nonumber \\ &e^{2\alpha_2}=e^{(\bar{k}_1-k_1)x_0+(\bar{l}_1-l_1)y_0+(\bar{\omega}_1-\omega_1)t_0+2\alpha_1}, \quad e^{2\alpha_4}=e^{(\bar{k}_1-k_1)x_0+(\bar{l}_3-l_3)y_0+(\bar{\omega}_1-\omega_1)t_0+2\alpha_3}.
\end{align}
Therefore, one-soliton solution of the complex reverse $xy$-space-time shifted nonlocal Maccari system \eqref{CMacxyt1}-\eqref{CMacxyt2} is
\begin{align}
&u(x,y,t)=\frac{e^{k_1x+l_1y+\omega_1t+\alpha_1}}{1+\frac{e^{(k_1+\bar{k}_1)x+(l_1+\bar{l}_1)y+(\omega_1+\bar{\omega}_1)t}}{(k_1+\bar{k}_1)(l_1+\bar{l}_1)}\big(\sigma_1e^{\alpha_1+\bar{\alpha}_1}+\sigma_3e^{\alpha_3+\bar{\alpha}_3}\big)}, \nonumber \\
&w(x,y,t)=\frac{e^{k_1x+l_3y+\omega_1t+\alpha_3}}{1+\frac{e^{(k_1+\bar{k}_1)x+(l_1+\bar{l}_1)y+(\omega_1+\bar{\omega}_1)t}}{(k_1+\bar{k}_1)(l_1+\bar{l}_1)}\big(\sigma_1e^{\alpha_1+\bar{\alpha}_1}+\sigma_3e^{\alpha_3+\bar{\alpha}_3}\big)}, \nonumber\\
&p(x,y,t)=\frac{2\frac{(k_1+\bar{k}_1)}{(l_1+\bar{l}_1)}e^{(k_1+\bar{k}_1)x+(l_1+\bar{l}_1)y+(\omega_1+\bar{\omega}_1)t}\bigg(\sigma_1 e^{\alpha_1+\bar{\alpha}_1}+\sigma_3 e^{\alpha_3+\bar{\alpha}_3}\bigg)}{\bigg[1+\frac{e^{(k_1+\bar{k}_1)x+(l_1+\bar{l}_1)y+(\omega_1+\bar{\omega}_1)t}}{(k_1+\bar{k}_1)(l_1+\bar{l}_1)}\big(\sigma_1e^{\alpha_1+\bar{\alpha}_1}+\sigma_3e^{\alpha_3+\bar{\alpha}_3}\big)\bigg]^2}.\nonumber\\
&\label{sol-compxyt}
\end{align}
\smallskip

\noindent By applying Type 2, we get the case $(\varepsilon_1,\varepsilon_2,\varepsilon_3)=(-1,1,-1)$ but the solution is trivial.
\smallskip

Let $k_1=\delta_1+i\beta_1,\,  l_1=\delta_2+i\beta_2, \, \omega_1=\delta_3+i\beta_3, \, l_3=\delta_4+i\beta_4,\, e^{\alpha_1}=\delta_5+i\beta_5,\, e^{\alpha_3}=\delta_6+i\beta_6$, and $a=i\delta_7$ in the solution \eqref{sol-compxyt}. We have $\delta_3=-\frac{2\delta_1\beta_1}{\delta_7}$, $\beta_3=\frac{\delta_1^2-\beta_1^2}{\delta_7}$ from the dispersion relation $\omega_1=-\frac{k^2_1}{a}$, and $\delta_2=\delta_4$ from ${l_1+\bar{l}_1=l_3+\bar{l}_3}$. Then, the solution \eqref{sol-compxyt} of the complex reverse $xy$-space-time shifted nonlocal Maccari system \eqref{CMacxyt1}-\eqref{CMacxyt2} becomes
\begin{align}\label{xi3}
&|u|^2 = \frac{e^{2\delta_1x+2\delta_2y+2\delta_3t}(\delta_5^2+\beta_5^2)}{\Big(1
+\frac{e^{2\delta_1x+2\delta_2y+2\delta_3t}}{4\delta_1\delta_2}[\sigma_1(\delta_5^2+\beta_5^2)+\sigma_3(\delta_6^2+\beta_6^2)]\Big)^2}
,\nonumber\\
& |w|^2 = \frac{e^{2\delta_1x+2\delta_2y+2\delta_3t}(\delta_6^2+\beta_6^2)}{\Big(1+\frac{e^{2\delta_1x+2\delta_2y+2\delta_3t}}{4\delta_1\delta_2}
[\sigma_1(\delta_5^2+\beta_5^2)+\sigma_3(\delta_6^2+\beta_6^2)]\Big)^2} 
, \nonumber\\
&p=\frac{4\delta_1e^{2\delta_1x+2\delta_2y+2\delta_3t}[\sigma_1(\delta_5^2+\beta_5^2)
+\sigma_3(\delta_6^2+\beta_6^2)]}{\Big(1+\frac{e^{2\delta_1x+2\delta_2y+2\delta_3t}}{4\delta_1\delta_2}
[\sigma_1(\delta_5^2+\beta_5^2)+\sigma_3(\delta_6^2+\beta_6^2)]\Big)^2}.\nonumber\\&
\end{align}
The above solution is nonsingular if $\frac{\sigma_1(\delta_5^2+\beta_5^2)+\sigma_3(\delta_6^2+\beta_6^2)}{\delta_1\delta_2}>0$. Let us give a particular example.

\begin{ex}
Take the parameters of the solution \eqref{xi3} as $\sigma_1=1,\sigma_3=-1,\delta_1=1,\beta_1=-1,\delta_7=1,\delta_2=\frac{1}{2},\delta_5=\beta_5=\sqrt{2},\delta_6=\beta_6=1$. We have one-soliton solution of the complex reverse $xy$-space-time shifted nonlocal Maccari system \eqref{CMacxyt1}-\eqref{CMacxyt2} as
\begin{equation}
|u|^2=\frac{4e^{2x+y+4t}}{\Big(1+e^{2x+y+4t}\Big)^2},\quad |w|^2=\frac{2e^{2x+y+4t}}{\Big(1+e^{2x+y+4t}\Big)^2},\quad p=\frac{8e^{2x+y+4t}}{\Big(1+e^{2x+y+4t}\Big)^2}.
\end{equation}
Here, the forms of the solutions $|u|^2$, $|w|^2$, and $p$ are similar, so only $|u|^2$ is plotted in Figure 1.
\begin{center}
\begin{figure}[h]
\centering
\begin{minipage}[t]{1\linewidth}
\centering
\includegraphics[angle=0,scale=.47]{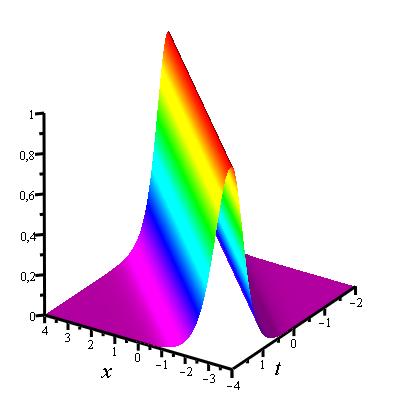}
\caption{Soliton solution $|u|^2$ of the system \eqref{CMacxyt1}-\eqref{CMacxyt2} for $y=0$.}
\end{minipage}%
\end{figure}
\end{center}
\end{ex}
\squeezeup

\section{Soliton solutions of the reduced shifted nonlocal Maccari equations}

In Section 5, we have obtained one-soliton solutions for the reduced shifted nonlocal Maccari systems. Now, we shall use these solutions with the reductions
\begin{align}
    &1)\, w=\rho_3u(\mu_1x+X_0,\mu_2y+Y_0,t),\quad \rho_3^2=\mu_1^2=\mu_2^2=1,\quad X_0, Y_0\in \mathbb{R}, \label{mui1} \\
    &2)\, w=\rho_3\bar{u}(\mu_1x+X_0,\mu_2y+Y_0,-t+T_0),\quad \rho_3^2=\mu_1^2=\mu_2^2=1,\quad X_0, Y_0, T_0\in \mathbb{R}, \label{mui2}
  \end{align}
to obtain solutions of the reduced shifted nonlocal Maccari equations.

  \subsection{One-soliton solutions of the shifted nonlocal equations reduced from the real $y$-space Maccari system \eqref{RMacy1}-\eqref{RMacy2}}

Let us use Type 1 approach with the reduction formula \eqref{mui1} and the solution (\ref{sol-realy}) of the real $y$-space Maccari system \eqref{RMacy1}-\eqref{RMacy2}. We obtain
	\begin{align}
	k_1=\mu_1 k_1, \quad l_3=\mu_2 l_1, \quad e^{\alpha_3}=\rho_3 e^{k_1X_0+l_1Y_0+\alpha_1}.
	\end{align}
It follows that
\begin{align}
    &a)\, \mu_1=1,\, X_0=0,\, \mu_2=1,\, Y_0=0, \\
    &b)\, \mu_1=1,\, X_0=0,\, \mu_2=-1.
\end{align}
Both of these constraints give the Maccari equation \eqref{RRy1}. We get the following constraints for $a)$ and $b)$, respectively as
$$l_3=l_1, \quad Y_0=0, \quad e^{\alpha_3}=\rho_3e^{\alpha_1},$$ 
and
$$l_3=-l_1,\quad y_0=Y_0, \quad e^{\alpha_3}=\rho_3e^{l_1y_0+\alpha_1},$$ 
 under the condition \eqref{A1con}. Hence one-soliton solution of the shifted nonlocal Maccari equation \eqref{RRy1} is
\begin{align}\label{solnmac1}
&u=\frac{e^{k_1x+l_1y+\omega_1t+\alpha_1}}{1+\sigma_1\frac{e^{(k_1+\bar{k}_1)x+(\omega_1+\bar{\omega}_1)t+\alpha_1+\bar{\alpha}_1}}{(k_1+\bar{k}_1)(l_1+\bar{l}_1)}
\Big( e^{(l_1+\bar{l}_1)y}+e^{(l_1+\bar{l}_1)(-y+y_0)}\Big)},  \\
&p=\frac{\sigma_1\frac{2(k_1+\bar{k}_1)}{(l_1+\bar{l}_1)}e^{(k_1+\bar{k}_1)x+(\omega_1+\bar{\omega}_1)t+\alpha_1+\bar{\alpha}_1}\Big( e^{(l_1+\bar{l}_1)y}+e^{(l_1+\bar{l}_1)(-y+y_0)}\Big)}{\Big[{1+\sigma_1\frac{e^{(k_1+\bar{k}_1)x+(\omega_1+\bar{\omega}_1)t+\alpha_1+\bar{\alpha}_1}}{(k_1+\bar{k}_1)(l_1
+\bar{l}_1)}\Big( e^{(l_1+\bar{l}_1)y}+e^{(l_1+\bar{l}_1)(-y+y_0)}\Big)}\Big]^2}.
\end{align}
Note that if we use Type 2 with the reduction \eqref{mui1}, we obtain trivial solution.

Now we use the reduction \eqref{mui2} with the one-soliton solution (\ref{sol-realy}). Using Type 1 gives
	\begin{align}
	k_1=\mu_1 \bar{k}_1, \quad l_3=\mu_2 \bar{l}_1, \quad \omega_1=-\bar{\omega}_1,
\quad e^{\alpha_3}=\rho_3 e^{\bar{k}_1X_0+\bar{l}_1Y_0+\bar{\omega}_1T_0+\bar{\alpha}_1}.
	\end{align}
It follows that
\begin{align}
    &a)\, \mu_1=1,\, X_0=0, \, \mu_2=1,\, Y_0=0, \\
    &b)\, \mu_1=1,\, X_0=0, \, \mu_2=-1.
\end{align}
Both of these constraints give the shifted nonlocal Maccari equation \eqref{CRy3}. We get the following constraints for $a)$ and $b)$, respectively:
$$l_3=\bar{l}_1,\quad Y_0=0, \quad e^{\alpha_3}=\rho_3e^{\bar{\omega}_1T_0+\bar{\alpha}_1},$$ 
and
$$l_3=-\bar{l}_1,\quad y_0=Y_0, \quad e^{\alpha_3}=\rho_3e^{\bar{l}_1y_0+\bar{\omega}_1T_0+\bar{\alpha}_1},$$ 
 under the condition \eqref{B1con}. Thus, one-soliton solution of the shifted nonlocal Maccari equation \eqref{CRy3} is
{\small\begin{equation}\label{solnmac2}
u=\frac{e^{k_1x+l_1y+\omega_1t+\alpha_1}}{1+\sigma_1\frac{e^{2k_1x+\alpha_1+\bar{\alpha}_1}}{(2k_1)(l_1+\bar{l}_1)}\Big( e^{(l_1+\bar{l}_1)y}+e^{(l_1+\bar{l}_1)(-y+y_0)}\Big)},\,\, p=\frac{\sigma_1\frac{4k_1}{(l_1+\bar{l}_1)}e^{2k_1x+\alpha_1+\bar{\alpha}_1}\Big( e^{(l_1+\bar{l}_1)y}+e^{(l_1+\bar{l}_1)(-y+y_0)}\Big)}{\Big[{1+\sigma_1\frac{e^{2k_1x+\alpha_1+\bar{\alpha}_1}}{2k_1(l_1+\bar{l}_1)}\Big( e^{(l_1+\bar{l}_1)y}+e^{(l_1+\bar{l}_1)(-y+y_0)}\Big)}\Big]^2}.
\end{equation}}
Note that if we use Type 2 approach with the reduction \eqref{mui2} we get trivial solution.

 Let $k_1=\delta_1+i\beta_1, \,  l_1=\delta_2+i\beta_2, \,
\omega_1=\delta_3+i\beta_3, \, e^{\alpha_1}=\delta_4+i\beta_4$, and $a=i\beta_5$ in the solution \eqref{solnmac2}. We have $\delta_3=-\frac{2\delta_1\beta_1}{\beta_5}$, $\beta_3=\frac{\delta_1^2-\beta_1^2}{\beta_5}$ from the dispersion relation $\omega_1=-\frac{k^2_1}{a}$. Then, the solution \eqref{solnmac2} becomes
\begin{equation}\label{xi4}
|u|^2 = \frac{e^{2\delta_1x+2\delta_2y+2\delta_3t}(\delta_4^2+\beta_4^2)}{\bigg(1+\sigma_1\frac{(\delta_4^2+\beta_4^2)}{4\delta_2\sqrt{\delta_1^2+\beta_1^2}}e^{2\delta_1x}[e^{2\delta_2y}+e^{2\delta_2(-y+y_0)}]\bigg)^2}
, \,\, p=\frac{\sigma_1\frac{2\delta_1(\delta_4^2+\beta_4^2)}{\delta_2}e^{2\delta_1x}[e^{2\delta_2y}+e^{2\delta_2(-y+y_0)}]}{\bigg(1+\sigma_1\frac{(\delta_4^2+\beta_4^2)}{4\delta_1\delta_2}e^{2\delta_1x}[e^{2\delta_2y}+e^{2\delta_2(-y+y_0)}]\bigg)^2}.
\end{equation}
The above solution is singular for
\begin{align}
    2\beta_1x-\arctan\Big(\frac{\beta_1}{\delta_1}\Big)=n\pi, \quad \sigma_1\frac{(\delta_4^2+\beta_4^2)}{4\delta_2\sqrt{\delta_1^2+\beta_1^2}}e^{2\delta_1x}[e^{2\delta_2y}+e^{2\delta_2(-y+y_0)}]+(-1)^n=0,\quad n\in\mathbb{N}.
\end{align}

If $k_1 \in \mathbb{R}$ and $\frac{\sigma_1}{\delta_2}>0$, giving $\beta_1=0$, $\delta_3=0$ in (\ref{xi4}), we get nonsingular solution for  the shifted nonlocal  Maccari equation \eqref{CRy3} as
\begin{align}
|u|^2 = \frac{e^{2\delta_1x+2\delta_2y}(\delta_4^2+\beta_4^2)}{\bigg(1+\sigma_1\frac{(\delta_4^2+\beta_4^2)}{4\delta_2|\delta_1|}e^{2\delta_1x}[e^{2\delta_2y}+e^{2\delta_2(-y+y_0)}]\bigg)^2}, \quad &p=\frac{\sigma_1\frac{2\delta_1(\delta_4^2+\beta_4^2)}{\delta_2}e^{2\delta_1x}[e^{2\delta_2y}+e^{2\delta_2(-y+y_0)}]}{\bigg(1+\sigma_1\frac{(\delta_4^2+\beta_4^2)}{4\delta_1\delta_2}e^{2\delta_1x}[e^{2\delta_2y}+e^{2\delta_2(-y+y_0)}]\bigg)^2}.
\end{align}
Let us give a particular example.

\begin{ex}
Take the parameters of the solution \eqref{xi4} as $\sigma_1=-1,\delta_1=1,\beta_1=0,\delta_3=0,\delta_2=-1, \delta_4=\beta_4=\sqrt{2}, \beta_5=\frac{1}{2}, y_0=1$. We have one-soliton solution of the shifted nonlocal Maccari equation \eqref{CRy3} as
\begin{equation}
|u|^2=\frac{4e^{2x-2y}}{\Big(1+e^{2x-2y}+e^{2x-2y-2}\Big)^2},\quad p=\frac{8e^{2x}\big(e^{-2y}+e^{2y-2}\big)}{\Big(1+e^{2x-2y}+e^{2x-2y-2}\Big)^2}.
\end{equation}
The above solutions $|u|^2$ and $p$ are plotted in Figure 2 and Figure 3, respectively.
\end{ex}

\begin{center}
\begin{figure}[ht]
\centering
\begin{minipage}[t]{0.41\linewidth}
\centering
\includegraphics[angle=0,scale=.47]{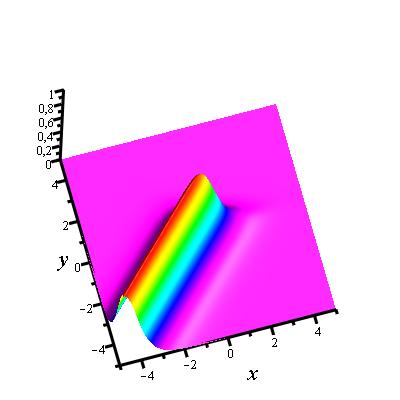}
\caption{Solitoff solution $|u|^2$ of \\ the equation \eqref{CRy3}}
\end{minipage}
\hspace{1.8cm}
\begin{minipage}[t]{0.41\linewidth}
\centering
\includegraphics[angle=0,scale=.47]{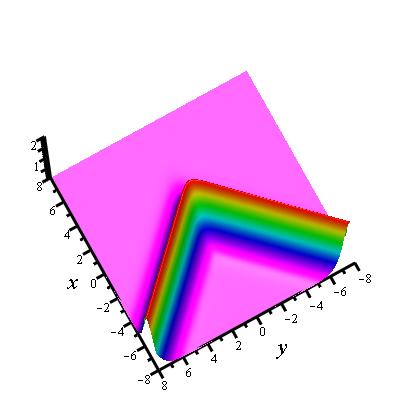}
\caption{V-type wave solution $p$ of \\ the equation \eqref{CRy3}}
\end{minipage}
\end{figure}
\end{center}

\subsection{One-soliton solution of the shifted nonlocal equations reduced from the real $xy$-space Maccari system \eqref{RMacxy1}-\eqref{RMacxy2}}

We first use Type 1 approach with the reduction formula \eqref{mui1} and one-soliton solution (\ref{sol-realxy}) of the system \eqref{RMacxy1}-\eqref{RMacxy2}. We obtain
\begin{align}
	k_1=\mu_1 k_1, \quad l_3=\mu_2 l_1, \quad e^{\alpha_3}=\rho_3 e^{k_1X_0+l_1Y_0+\alpha_1}.
	\end{align}
It follows that
\begin{align}
    &a)\, \mu_1=1,\, X_0=0, \mu_2=1,\, Y_0=0, \label{RRxy.a}\\
    &b)\, \mu_1=1,\, X_0=0, \mu_2=-1.   \label{RRxy.b}
\end{align}
The case \eqref{RRxy.a} gives the shifted nonlocal Maccari equation \eqref{RRxy1} with the constraints
$$l_3=l_1, \quad Y_0=0, \quad e^{\alpha_3}=\rho_3e^{\alpha_1},$$ 
and the condition \eqref{A1con}. We have the solution of the equation \eqref{RRxy1} as
\begin{equation}\label{solnmac3}
u=\frac{e^{k_1x+l_1y+\omega_1t+\alpha_1}}{1+\sigma_1\frac{e^{2k_1x+(l_1+\bar{l}_1)y+\alpha_1+\bar{\alpha}_1}}{k_1(l_1+\bar{l}_1)}},\quad p=\frac{\sigma_1\frac{8k_1}{(l_1+\bar{l}_1)}e^{2k_1x+(l_1+\bar{l}_1)y+\alpha_1+\bar{\alpha}_1}}{\Big[1+\sigma_1\frac{e^{2k_1x+(l_1+\bar{l}_1)y+\alpha_1+\bar{\alpha}_1}}
{k_1(l_1+\bar{l}_1)}\Big]^2}.
\end{equation}
The case \eqref{RRxy.b} gives trivial solution. Note that if we use Type 2 with the reduction \eqref{mui1} here, we also get trivial solution.

Let us now use Type 1 with the complex reduction \eqref{mui2} and one-soliton solution (\ref{sol-realxy}). We obtain
\begin{align}
	k_1=\mu_1 k_1, \quad l_3=\mu_2 \bar{l}_1, \quad \omega_1=-\bar{\omega}_1, \quad e^{\alpha_3}=\rho_3 e^{\bar{k}_1X_0+\bar{l}_1Y_0+\bar{\omega}_1T_0+\bar{\alpha}_1}. \hfill
\end{align}
It follows that
\begin{align}
    &a)\, \mu_1=1,\, X_0=0, \, \mu_2=1,\, Y_0=0, \label{RCxy.a1}\\
    &b)\, \mu_1=1,\, X_0=0, \, \mu_2=-1.   \label{RCxy.b1}
\end{align}
The case \eqref{RCxy.a1} gives the Maccari equation \eqref{CRxy3} with the constraints
$$l_3=\bar{l}_1, \quad Y_0=0, \quad e^{\alpha_3}=\rho_3e^{\bar{\omega}_1T_0+\bar{\alpha}_1},$$ 
 and the condition \eqref{B1con}. Hence one-soliton solution of the shifted nonlocal Maccari equation \eqref{CRxy3} is the solution \eqref{solnmac3}.
 The case \eqref{RCxy.b1} gives trivial solution.

We now use Type 2 approach with the reduction \eqref{mui2} and solution (\ref{sol-realxy}). We obtain the constraints
	\begin{align}
	-k_1=\mu_1 k_1, \quad l_3=-\mu_2 l_1, \quad \omega_1=-\bar{\omega}_1, \quad e^{\alpha_3}=\rho_3 e^{\bar{k}_1X_0+\bar{l}_1Y_0+\bar{\omega}_1T_0+\bar{\alpha}_1}.
	\end{align}

It follows that
\begin{align}
    &a)\, \mu_1=-1, \, \mu_2=1,\, Y_0=0, \label{RCxy.a2}\\
    &b)\, \mu_1=-1, \, \mu_2=-1.   \label{RCxy.b2}
\end{align}

The case \eqref{RCxy.a2} leads to the  Maccari equation \eqref{CRxy4} having trivial solution. Similarly, the case \eqref{RCxy.b2} leads to the Maccari equation \eqref{CRxy3} having the solution \eqref{solnmac3}.

\subsection{One-soliton solutions of the shifted nonlocal equations reduced from the complex reverse time Maccari system \eqref{CMact1}-\eqref{CMact2}}

Use Type 1 with the reduction formula \eqref{mui1} and one-soliton solution (\ref{sol-compt}) of the complex reverse time Maccari system \eqref{CMact1}-\eqref{CMact2}. We obtain
\begin{align}
	k_1=\mu_1 k_1, \quad l_3=\mu_2 l_1, \quad e^{\alpha_3}=\rho_3 e^{k_1X_0+l_1Y_0+\alpha_1},
	\end{align}
giving
\begin{align}
    &a)\, \mu_1=1,\, X_0=0,\, \mu_2=1,\, Y_0=0, \label{CRt.a1}\\
    &b)\, \mu_1=1,\, X_0=0,\, \mu_2=-1.   \label{CRt.b1}
\end{align}

The case \eqref{CRt.a1} gives the shifted nonlocal Maccari equation \eqref{RCt1} with the constraints
$$l_3=l_1, \quad Y_0=0, \quad e^{\alpha_3}=\rho_3e^{\alpha_1},$$
and the condition \eqref{A4con}. Here, we have the solution \eqref{solnmac3}.

The case \eqref{CRt.b1} gives the shifted nonlocal Maccari equation \eqref{RCt2} with the constraints
$$l_3=-l_1, \quad e^{\alpha_3}=\rho_3e^{l_1Y_0+\alpha_1},$$
and the condition \eqref{A4con}. Hence, we have the solution \eqref{solnmac2} here.

Let us now use Type 2 approach with the reduction \eqref{mui1} and one-soliton solution (\ref{sol-compt}). We get
\begin{align}
	k_1=-\mu_1 k_1, \quad l_3=-\mu_2 \bar{l}_1, \quad e^{2\alpha_3}=e^{(l_1-\bar{l}_1)Y_0+2\alpha_1},
	\end{align}
yielding
\begin{align}
    &a)\, \mu_1=-1,\, \mu_2=1,\, Y_0=0, \label{CRt.a2}\\
    &b)\, \mu_1=-1,\, \mu_2=-1.   \label{CRt.b2}
\end{align}
The case \eqref{CRt.a2} gives trivial solution. On the other hand, the case \eqref{CRt.b2} corresponds to the shifted nonlocal  Maccari equation \eqref{RCt4} with $l_3=\bar{l}_1$ and the condition \eqref{A4con}. Here, we again get the solution \eqref{solnmac3}.

We now use Type 1 approach with the reduction \eqref{mui2} and the solution (\ref{sol-compt}). We obtain
\begin{align}
	k_1=\mu_1 k_1, \quad l_3=\mu_2 \bar{l}_1, \quad \omega_1=-\bar{\omega}_1, \quad e^{\alpha_3}=\rho_3 e^{\bar{k}_1X_0+\bar{l}_1Y_0+\bar{\omega}_1T_0+\bar{\alpha}_1},
\end{align}
giving
\begin{align}
    &a)\, \mu_1=1,\, X_0=0, \, \mu_2=1,\, Y_0=0, \label{CCt.a1}\\
    &b)\, \mu_1=1,\, X_0=0, \, \mu_2=-1.   \label{CCt.b1}
\end{align}
The case \eqref{CCt.a1} corresponds to the shifted nonlocal Maccari equation \eqref{CCt1} with the constraints
$$l_3=\bar{l}_1, \quad Y_0=0, \quad e^{\alpha_3}=\rho_3e^{\bar{\omega}_1T_0+\bar{\alpha}_1},$$
 and the condition \eqref{B4con}. Note that the Maccari equations \eqref{CCt1} and \eqref{RCt1} are the same, i.e., the solution for the shifted nonlocal Maccari equation \eqref{CCt1} is given by \eqref{solnmac3}.

The case \eqref{CCt.b1} yields the Maccari equation \eqref{CCt2} with the constraints
$$l_3=-\bar{l}_1, \quad e^{\alpha_3}=\rho_3e^{\bar{l}_1Y_0+\bar{\omega}_1T_0+\bar{\alpha}_1},$$
and the condition \eqref{B4con}. Notice that the Maccari equations \eqref{CCt2} and \eqref{RCt2} are the same, that is the solution for the shifted nonlocal Maccari equation \eqref{CCt2} is given by \eqref{solnmac2}.

Now we use Type 2 with the reduction \eqref{mui2} and the solution (\ref{sol-compt}). We obtain
	\begin{align}
	-k_1=\mu_1 \bar{k}_1, \quad l_3=-\mu_2 l_1, \quad \omega_1=-\bar{\omega}_1, \quad e^{2\alpha_3}=e^{(\bar{l}_1-l_1)Y_0+2\bar{\omega}_1T_0+2\bar{\alpha}_1},
	\end{align}
yielding
\begin{align}
    &a)\, \mu_1=-1,\, \mu_2=1,\, Y_0=0, \label{CCt.a2}\\
    &b)\, \mu_1=-1,\, \mu_2=-1.   \label{CCt.b2}
\end{align}

The case \eqref{CCt.a2} yields trivial solution. On the other hand, the case \eqref{CCt.b2} corresponds to the shifted nonlocal  Maccari equation \eqref{CCt4} with the conditions $l_3=l_1$ and \eqref{B4con}. Note that the Maccari equations \eqref{CCt4} and \eqref{RCt4} are the same. Hence, the solution for the shifted nonlocal Maccari equation \eqref{CCt4} is given by \eqref{solnmac3}.

\subsection{One-soliton solutions of the shifted nonlocal equations reduced from the complex reverse $y$-space time Maccari system \eqref{CMacyt1}-\eqref{CMacyt2}}

We first use Type 1 approach with the reduction \eqref{mui1} and one-soliton solution (\ref{sol-compyt}) of the complex reverse $y$-space time Maccari system \eqref{CMacyt1}-\eqref{CMacyt2}. We get
\begin{align}
	k_1=\mu_1 k_1, \quad l_3=\mu_2 l_1, \quad e^{\alpha_3}=\rho_3 e^{k_1X_0+l_1Y_0+\alpha_1}, \hfill
	\end{align}
giving
\begin{align}
    &a)\, \mu_1=1,\, X_0=0, \mu_2=1,\, Y_0=0, \label{CRyt.a1}\\
    &b)\, \mu_1=1,\, X_0=0, \mu_2=-1.   \label{CRyt.b1}
\end{align}
The case \eqref{CRyt.a1}  corresponds to the Maccari equation \eqref{RCyt1} with the constraints
$l_3=l_1, Y_0=0, e^{\alpha_3}=\rho_3e^{\alpha_1}$, and the condition \eqref{A4con}. On the other hand, the case \eqref{CRyt.b1} gives the equation \eqref{RCyt2} with the constraints $l_3=-l_1, e^{\alpha_3}=\rho_3e^{l_1Y_0+\alpha_1}$, and the condition \eqref{A4con}. Here, both of these equations have the same solution \eqref{solnmac2}.

Note that if we use Type 2 with the reduction \eqref{mui1} and the solution (\ref{sol-compyt}) we get trivial solution.

We now use Type 1 with the complex reduction \eqref{mui2} and one-soliton solution (\ref{sol-compyt}). We have
\begin{align}
	k_1=\mu_1 k_1, \quad l_3=\mu_2 \bar{l}_1, \quad \omega_1=-\bar{\omega}_1, \quad e^{\alpha_3}=\rho_3 e^{\bar{k}_1X_0+\bar{l}_1Y_0+\bar{\omega}_1T_0+\bar{\alpha}_1},
\end{align}
giving
\begin{align}
    &a)\, \mu_1=1,\, X_0=0, \, \mu_2=1,\, Y_0=0, \label{CCyt.a2}\\
    &b)\, \mu_1=1,\, X_0=0, \, \mu_2=-1.   \label{CCyt.b2}
\end{align}
Under the condition \eqref{B4con}, the cases \eqref{CCyt.a2} and \eqref{CCyt.b2} give the shifted nonlocal Maccari equations \eqref{CCyt2} and \eqref{CCyt1}, respectively. For the equation \eqref{CCyt2} we have the constraints
$l_3=\bar{l}_1, Y_0=0,  e^{\alpha_3}=\rho_3e^{\bar{\omega}_1T_0+\bar{\alpha}_1}$, while for the equation \eqref{CCyt1} the constraints are $l_3=-\bar{l}_1, e^{\alpha_3}=\rho_3e^{\bar{l}_1Y_0+\bar{\omega}_1T_0+\bar{\alpha}_1}$. We get the same solution \eqref{solnmac2} for both equations.

Note that if we use Type 2 approach with the reduction \eqref{mui2} and solution (\ref{sol-compyt}) we obtain trivial solution.

\subsection{One-soliton solutions of the shifted nonlocal equations reduced from the complex reverse $xy$-space time Maccari system \eqref{CMacxyt1}-\eqref{CMacxyt2}}

When we use Type 1 with the reduction \eqref{mui1} and one-soliton solution (\ref{sol-compxyt}) of the shifted nonlocal Maccari system \eqref{CMacxyt1}-\eqref{CMacxyt2}, we obtain
\begin{align}
	k_1=\mu_1 k_1, \quad l_3=\mu_2 l_1, \quad e^{\alpha_3}=\rho_3 e^{k_1X_0+l_1Y_0+\alpha_1},
	\end{align}
giving
\begin{align}
    &a)\, \mu_1=1,\, X_0=0,\, \mu_2=1,\, Y_0=0, \label{CRxyt.a1}\\
    &b)\, \mu_1=1,\, X_0=0,\, \mu_2=-1.   \label{CRxyt.b1}
\end{align}

The case \eqref{CRxyt.a1} corresponds to the shifted nonlocal  Maccari equation \eqref{RCxyt1} with the constraints
$$l_3=l_1, \quad Y_0=0, \quad e^{\alpha_3}=\rho_3e^{\alpha_1},$$
and the condition \eqref{A4con}. Hence the solution of the shifted nonlocal equation \eqref{RCxyt1} is
\begin{equation}\label{solnmac4}
u=\frac{e^{k_1x+l_1y+\omega_1t+\alpha_1}}{1+\sigma_1\frac{2e^{(k_1+\bar{k}_1)x+(l_1+\bar{l}_1)y+(\omega_1+\bar{\omega}_1)t+\alpha_1+\bar{\alpha}_1}}{(k_1+\bar{k}_1)
(l_1+\bar{l}_1)}},  \,\,
p=\frac{\sigma_1\frac{4(k_1+\bar{k}_1)}{(l_1+\bar{l}_1)}e^{(k_1+\bar{k}_1)x+(l_1+\bar{l}_1)y+(\omega_1+\bar{\omega}_1)t+\alpha_1+\bar{\alpha}_1}}{\Big[{1
+\sigma_1\frac{2e^{(k_1+\bar{k}_1)x+(l_1+\bar{l}_1)y+(\omega_1+\bar{\omega}_1)t+\alpha_1+\bar{\alpha}_1}}{(k_1+\bar{k}_1)(l_1+\bar{l}_1)}}\Big]^2}.
\end{equation}

The case \eqref{CRxyt.b1} gives trivial solution.

\noindent Using Type 2 approach with the reduction formula \eqref{mui1} and solution (\ref{sol-compt}) we have
\begin{align}
	k_1=-\mu_1 \bar{k}_1, \quad l_3=-\mu_2 \bar{l}_1, \quad \omega_1=-\bar{\omega}_1, \quad e^{2\alpha_3}= e^{(l_1-\bar{l}_1)Y_0+2\alpha_1}.
	\end{align}
It follows that
\begin{align}
    &a)\, \mu_1=-1,\, \mu_2=1,\, Y_0=0, \label{CRxyt.a2}\\
    &b)\, \mu_1=-1,\, \mu_2=-1.   \label{CRxyt.b2}
\end{align}

The case \eqref{CRxyt.a2} gives trivial solution. On the other hand, the case \eqref{CRxyt.b2} corresponds to the shifted nonlocal  Maccari equation \eqref{RCxyt4} with $l_3=\bar{l}_1$ and the condition \eqref{A4con}. Thus, we again obtain the solution \eqref{solnmac3}.

We now use Type 1 with the reduction formula \eqref{mui2} and solution (\ref{sol-compxyt}). We obtain
\begin{align}
	k_1=\mu_1 \bar{k}_1, \quad l_3=\mu_2 \bar{l}_1, \quad \omega_1=-\bar{\omega}_1, \quad e^{\alpha_3}=\rho_3 e^{\bar{k}_1X_0+\bar{l}_1Y_0+\bar{\omega}_1T_0+\bar{\alpha}_1},
\end{align}
yielding
\begin{align}
    &a)\, \mu_1=1,\, X_0=0, \, \mu_2=1,\, Y_0=0, \label{CCxyt.a1}\\
    &b)\, \mu_1=1,\, X_0=0, \, \mu_2=-1.\label{CCxyt.b1}
\end{align}
The case \eqref{CCxyt.a1} corresponds to the shifted nonlocal  Maccari equation \eqref{CCxyt4} with
$l_3=\bar{l}_1, e^{\alpha_3}=\rho_3e^{\bar{\omega}_1T_0+\bar{\alpha}_1}$, and the condition \eqref{B4con}. Hence, we obtain the solution \eqref{solnmac3} for the equation \eqref{CCxyt4}. The case \eqref{CCxyt.b1} gives trivial solution.

Using Type 2 with the reduction forumula \eqref{mui2} and one-soliton solution (\ref{sol-compxyt}) we get
\begin{align}
	k_1=-\mu_1 k_1, \quad l_3=-\mu_2 l_1, \quad e^{2\alpha_3}=e^{(k_1-\bar{k}_1)X_0+(l_1-\bar{l}_1)Y_0+(\omega_1-\bar{\omega}_1)T_0+2\bar{\alpha}_1},
\end{align}
giving
\begin{align}
    &a)\, \mu_1=-1, \, \mu_2=1,\, Y_0=0, \label{CCxyt.a2}\\
    &b)\, \mu_1=-1, \mu_2=-1.   \label{CCxyt.b2}
\end{align}
The case \eqref{CCxyt.a2} gives trivial solution. On the other hand, the case \eqref{CCxyt.b2} corresponds to the Maccari equation \eqref{CCxyt1} with the constraints
\begin{equation}
l_3=l_1, \quad e^{2\alpha_3}=e^{(k_1-\bar{k}_1)X_0+(l_1-\bar{l}_1)Y_0+(\omega_1-\bar{\omega}_1)T_0+2\bar{\alpha}_1},
\end{equation}
and the condition \eqref{B4con}. Therefore, we obtain the solution \eqref{solnmac4} for the Maccari equation \eqref{CCxyt1}.

\section{Concluding remarks}
In this work, we studied shifted nonlocal reductions of $(2+1)$-dimensional $5$-component Maccari system. We first proved that shifted nonlocal reductions are special cases of shifted scale transformations.

We then obtained all consistent shifted nonlocal reductions of this system. By the real shifted nonlocal reduction, we derived real reverse $x$-space, $y$-space, and $xy$-space shifted nonlocal Maccari systems, which are all two-place systems. Using complex shifted nonlocal reductions on the $5$-component Maccari system we obtained the same systems derived by the real shifted nonlocal reductions in addition to complex reverse time, $x$-space time, $y$-space time, and $xy$-space time shifted nonlocal two-place Maccari systems.

We further reduced the obtained shifted nonlocal Maccari systems by applying real and complex shifted nonlocal reductions. By doing so, we derived $3$ two-place and $9$ four-place shifted nonlocal Maccari equations from the real shifted nonlocal Maccari systems. On the other hand, we obtained $4$ two-place and $12$ four-place different complex shifted nonlocal Maccari equations from the complex shifted nonlocal Maccari systems. In total, we obtained new integrable $28$ shifted nonlocal Maccari equations, consisting of $7$ two-place and $21$ four-place equations.

We also obtained one-soliton solution of the $5$-component Maccari system by Hirota method. Using this solution with the reduction formulas gave conditions on the parameters of the solutions of the reduced Maccari systems and equations. Under these conditions, we derived soliton solutions (bell-type, solitoff, V-type) of the shifted nonlocal  Maccari systems and equations.

As a future work, we plan to obtain different types of exact solutions to shifted nonlocal
reductions of $5$-component Maccari system. We want to generalize multi-place nonlocality for the $(N+1)$-component Maccari system.

\section{Acknowledgment}
This work is partially supported by the Scientific
and Technological Research Council of Turkiye (T\"{U}B\.{I}TAK).


\begin{thebibliography}{lll}

\bibitem{abl1} M.J. Ablowitz and Z.H. Musslimani, Integrable nonlocal nonlinear Schr\"{o}dinger equation, Phys. Rev. Lett. \textbf{110}, 064105, 2013.

\bibitem{abl2} M.J. Ablowitz and Z.H. Musslimani,  Inverse scattering transform for the integrable nonlocal nonlinear Schr\"{o}dinger equation, Nonlinearity \textbf{29}, 915--946, 2016.

\bibitem{abl3} M.J. Ablowitz and Z.H. Musslimani, Integrable nonlocal nonlinear equations, Stud. Appl. Math. \textbf{139} (1), 7--59, 2016.

\bibitem{AbMu4} M.J. Ablowitz and Z.H. Musslimani, Integrable space-time shifted nonlocal nonlinear equations, Phys. Lett. A \textbf{409}, 127516, 2021.

\bibitem{AbMu5} M.J. Ablowitz, Z.H. Musslimani, and N.J. Ossi, Inverse scattering transform for continuous and discrete space-time shifted integrable equations, Stud. Appl. Math. \textbf{153}, e12764, 2024.

\bibitem{gur1}  M. G\"{u}rses and A. Pekcan, Nonlocal nonlinear Schr\"{o}dinger equations and their soliton solutions, J. Math. Phys. \textbf{59}, 051501, 2018.

\bibitem{gur3} M. G\"{u}rses and A. Pekcan, Nonlocal nonlinear modified KdV equations and their soliton solutions,
Commun. Nonlinear Sci. Numer. Simulat. \textbf{67}, 427--448, 2019.

\bibitem{gur4} M. G\"{u}rses and A. Pekcan, Nonlocal KdV equations, Phys. Lett. A {\bf 384} (35), 126894, 2020.

\bibitem{pek2021} A. Pekcan, Local and nonlocal (2+1)-dimensional Maccari systems and their soliton solutions, Phys. Scr. {\bf 96} (3), 035217, 2021.

\bibitem{gur5}  M. G\"{u}rses and A. Pekcan, Soliton solutions of the shifted nonlocal NLS and MKdV equations, Phys. Lett. A {\bf 422}, 127793, 2022.



\bibitem{SYLou1} S.Y. Lou, F. Huang, Alice-Bob physics: Coherent solutions of nonlocal KdV systems, Sci. Rep. {\bf 7}, 869, 2017.

\bibitem{SYLou2} C. Li, S.Y. Lou, M. Jia, Coherent structure of Alice–Bob modified Korteweg de-Vries equation, Nonlinear Dyn. {\bf 93}, 1799--1808, 2018.


\bibitem{KdVMKdVLou2018} S.Y. Lou, Alice-Bob systems, $\hat{P}-\hat{T}-\hat{C}$ symmetry invariant and symmetry breaking soliton solutions, J. Math. Phys. \textbf{59} (8), 083507, 2018.

\bibitem{MaMKdV(1)2022} W.X. Ma, Type $(-\lambda,-\lambda^{*})$ reduced nonlocal integrable mKdV equations and their soliton solutions, Appl. Math. Lett. \textbf{131}, 108074,
2022.

\bibitem{JiZhu2017} J.L. Ji, Z.N. Zhu, Soliton solution of an integrable nonlocal modified Korteweg-de Vries equation through inverse scattering transform, J. Math. Anal. Appl.
\textbf{453} (2), 973--984, 2017.

\bibitem{Khare2024} A. Khare, A. Saxena, New solutions of nonlocal NLS, mKdV and Hirota equations, Ann. Phys. \textbf{460}, 169561, 2024.


\bibitem{MaMKdV2021} L. Ling, W.X. Ma, Inverse scattering and soliton solutions of nonlocal complex reverse-spacetime modified Korteweg-de Vries hierarchies, Symmetry \textbf{13}, 512, 2021.



\bibitem{shiftedKdV22018} M. Jia, S.Y. Lou, Exact $P_sT_d$ invariant and $P_sT_d$ symmetric breaking solutions, symmetry reductions and B\"{a}cklund transformations for
an AB-KdV system, Phys. Lett. A \textbf{382} (17), 1157--1166, 2018.

\bibitem{LWZ} S.M. Liu, J. Wang, and D.J. Zhang,  Solutions to integrable space-time shifted nonlocal equations, Rep. Math. Phys. {\bf 89} (2), 199--220, 2022.

\bibitem{WW} X. Wang and J. Wei, Three types of Darboux transformation and general soliton solutions for the space-shifted nonlocal PT symmetric nonlinear Schr\"{o}dinger equation, Appl. Math. Lett. {\bf 130}, 107998, 2022.

\bibitem{shiftedKdV12022} X.Z. Liu, J. Yu, A study on a nonlocal coupled KdV system, Nonlinear Dyn. \textbf{108}, 569--577, 2022.

\bibitem{shiftedMKdV22022} J.P. Wu, Reduction approach and three types of multi-soliton solutions of the shifted nonlocal mKdV equation, Nonlinear Dyn. \textbf{109}, 3017-3027, 2022.

\bibitem{Maccari1997} A. Maccari, Universal and integrable nonlinear evolution systems of equations in (2+1) dimensions, J. Math. Phys. \textbf{38}, 4151, 1997.

\bibitem{Maccari2020} A. Maccari, The Maccari system as model system for rogue waves, Phys. Lett. A \textbf{384}, 126740, 2020.

\bibitem{Uthayakumar} A. Uthayakumar, K. Nakkeeran, and K. Porsezian, Soliton solutions of new (2+1) dimensional nonlinear partial differential equations,
Chaos Solitons Fractals \textbf{10} (9), 1513--1518, 1999.


\bibitem{HanChen1} Z. Han, Y. Chen, Bright-dark mixed N-soliton solution of the two-dimensional Maccari system, Chin Phys. Lett. \textbf{34} (7), 070202, 2017.

\bibitem{HanChen2} Z. Han, Y. Chen, Bright-dark mixed N-soliton solution of the two-dimensional multicomponent Maccari system, Z. Naturforsch. A \textbf{72} (8), 2017.

\bibitem{Xu} T. Xu, Y. Chen, and Z. Qiao, Multi-dark soliton solutions for the (2+1)-dimensional
multi-component Maccari system, Modern Phys. Lett. B  \textbf{33} (31), 1950390, 2019.

\bibitem{Wazwaz} W. Liu, A.M. Wazwaz, Dynamics of fusion and fission collisions between lumps and line solitons in the Maccari's system, Rom. J. Phys. \textbf{64}, 111, 2019.

\bibitem{HZZ} W.H. Huang, C.L. Zheng, and J.F. Zhang, Coherent soliton structures of a new (2+1) dimensional evolution equation, Acta Phys. Sin. (in Chinese) \textbf{51}, 2676, 2002.

\bibitem{HLM} W.H. Huang, Y.L. Liu, and Z.Y. Ma, Doubly periodic propagating wave patterns of (2+1)-dimensional Maccari system, Commun. Theor. Phys. \textbf{47}, 397--402, 2007.

\bibitem{HuangZhang} W.H. Huang, J.F. Zhang, Folded localized excitations of the Maccari system, Acta Phys. Polon. B \textbf{35} (8), 2004.

\bibitem{CCS} N. Cheemaa, S. Chen, and A.R. Seadawy, Propagation of isolated waves of coupled nonlinear (2+1)-dimensional Maccari system in plasma physics, Results Phys. \textbf{17}, 102987, 2020.

\bibitem{Yuan} F. Yuan, J. Rao, K. Porsezian, D. Mihalache, and J. He, Various exact rational solutions of the two-dimensional Maccari's system, Rom. J. Phys. \textbf{61} (3-4), 378--399, 2016.

\bibitem{Zhang} J.F. Zhang, Generalized dromion structures of new (2+1)-dimensional nonlinear evolution equation, Commun. Theor. Phys. \textbf{35}, 267--270, 2001.

\bibitem{ZMH} J.F. Zhang, J.P. Meng, and W.H. Huang, Novel class of coherent localized structures for the Maccari system, Commun. Theor. Phys. \textbf{40}, 443--446, 2003.

\bibitem{Loumulti} S.Y. Lou. Multi-place physics and multi-place nonlocal systems.
Commun. Theoret. Phys. \textbf{72}, 057001, 2020.

\bibitem{Mamulti} W.X. Ma, Soliton hierarchies and soliton solutions of type $(-\lambda^{\star},-\lambda)$ reduced
nonlocal nonlinear Schr\"{o}dinger equations of arbitrary even order, Partial Dif. Equ. Appl. Math. \textbf{7}, 100515, 2023.

\bibitem{GPZ} M. G\"{u}rses, A. Pekcan, and K. Zheltukhin, Discrete symmetries and nonlocal reductions, Phys. Lett. A \textbf{384}, 120065, 2020.

\bibitem{Hirota2} R. Hirota, The Direct Method in Soliton Theory, Cambridge University Press, Cambridge, 2004.

\bibitem{Hietarinta} J. Hietarinta, Introduction to the Bilinear Method, in: Integrability of
Nonlinear Systems, eds. Y. Kosman-Schwarzbach, B. Grammaticos, and K.M. Tamizhmani, Springer Lecture Notes in Physics \textbf{495}, 95--103, 1997.















\end{thebibliography}
\end{document}